%% file: TCOM_Vehicle_Rev_arxiv.tex
\documentclass[12pt, draftclsnofoot, onecolumn]{IEEEtran}
 \usepackage{amsmath,amssymb}
 \usepackage{subfigure}
 \usepackage{graphicx,graphics,color,psfrag}
 \usepackage{cite,balance}
 \usepackage{caption}
 \allowdisplaybreaks
 \usepackage{algorithm}
 \usepackage{accents}
 \usepackage{amsthm}
 \usepackage{bm}
 \usepackage{algorithmic}
 \usepackage[english]{babel}
 \usepackage{multirow}
 \usepackage{enumerate}
 \usepackage{cases}
 \usepackage{stfloats}
 \usepackage{dsfont}
 \usepackage{color,soul}
 \usepackage{amsfonts}
 \usepackage{cite,graphicx,amsmath,amssymb}
 \usepackage{subfigure}
  \usepackage{booktabs}
 \usepackage{fancyhdr}
 \usepackage{hhline}
 \usepackage{graphicx,graphics}
 \usepackage{array,color}

 \usepackage{amsmath}


\newcounter{parentnumber}

\include{header}

\usepackage{enumitem}
\setlist[itemize]{leftmargin=*}

\begin{document}

\title{\LARGE Hidden Vehicle Sensing via Asynchronous V2V Transmission: A Multi-Path-Geometry Approach}
\author{ Kaifeng Han, Seung-Woo Ko, Hyukjin Chae, Byoung-Hoon Kim, and \\ Kaibin Huang
\thanks{ K. Han and K. Huang are with The  University of  Hong Kong, Hong Kong. S.-W Ko is with Korea Maritime and Ocean University, S. Korea. H. Chae and B.-H Kim are with LG Electronics, S. Korea (email:
huangkb@eee.hku.hk). Part of this work will be presented in IEEE VTC 2018-Fall.}}
\maketitle

\vspace{-55pt}
\begin{abstract}
\vspace{-15pt}
Accurate vehicular sensing is a basic and important operation in autonomous driving. Unfortunately, the existing techniques have their own limitations. For instance,  the communication-based approach (e.g., transmission of GPS information) has high latency and low reliability while the reflection-based approach (e.g., RADAR) is incapable of detecting \emph{hidden vehicles} (HVs) without line-of-sight. This is arguably the reason behind some recent fatal accidents involving autonomous vehicles. To address this issue, this paper presents a novel HV-sensing technology that exploits multi-path transmission from a HV to a \emph{sensing vehicle} (SV). The powerful  technology enables the SV  to detect multiple HV-state parameters including position, orientation of driving direction, and size. Its implementation does not even require \emph{transmitter-receiver synchronization} like  conventional mobile positioning techniques. Our design approach leverages  estimated  information on multi-path [namely their \emph{angles-of-arrival} (AoA), \emph{angles-of-departure} (AoD), and \emph{time-of-arrival} (ToA)] and their geometric relations. As a result, a complex system of equations or  optimization problems, where   the desired HV-state parameters are  variables,  can be formulated for different channel-noise conditions.   The development of intelligent  solution methods ranging from least-square estimator to disk/box minimization yields a set of practical HV-sensing techniques. We  study their feasibility conditions in terms of the required number of paths. Furthermore, practical solutions, including sequential path combining and random directional beamforming, are proposed to enable HV-sensing given  insufficient paths. Last, realistic simulation of driving in both  highway and rural scenarios demonstrates the effectiveness of the proposed techniques. In summary, the proposed technique will enhance the capabilities of existing vehicular sensing technologies (e.g., RADAR and LIDAR) by enabling HV-sensing.
\end{abstract}
\vspace{-20pt}

\section{Introduction}

\emph{Autonomous driving} (auto-driving)  aims at   reducing  car accidents, traffic congestion, and greenhouse gas emissions by automating the  transportation process \cite{choi2016millimeter2}. The potential  impact of the cross-disciplinary technology has attracted heavy R\&D investments not only by leading car manufacturers (e.g., Tesla) but also Internet companies (e.g., Google). One primary operation  of auto-driving is vehicular positioning \cite{cui2016vehicle, cohda2016, balico2018localization}, namely positioning nearby vehicles and tracking  other parameters such as sizes and trajectories. The information then serves as inputs for computing and control tasks such as navigation and accidence avoidance. There exist diversified approaches for vehicular positioning but they all have their own drawbacks.  One conventional approach  is to exchange the absolute location information among a group of nearby vehicles by using \emph{vehicle-to-vehicle} (V2V) transmission \cite{yao2011improving2}. Its  main drawbacks include high latency due to the  data exchanging process and low reliability arising from inaccurate  \emph{Global Positioning System} (GPS) information in e.g., dense urban areas  or tunnels  \cite{choi2016millimeter2}. Another approach is to deploy sensors ranging from    \emph{LIght Detection and Ranging} (LIDAR) to \emph{RAdio Detection And Ranging} (RADAR). As illustrated in  Fig.~\ref{v2v_network}, they are  capable of  sensing the \emph{line-of-sight} (LoS) vehicles,
but cannot ``see through" a large solid object (e.g.,  a truck) to detect \emph{hidden vehicles} (HVs). A comprehensive discussion of existing approaches is given in the sequel.
Motivated by their drawbacks,  this paper presents a novel technology  for accurately sensing a HV including detecting its  position,  orientation of driving direction, and size by exploiting the multi-path geometry of asynchronous V2V transmission. 

\begin{figure}[t]
\centering
\includegraphics[width=8cm]{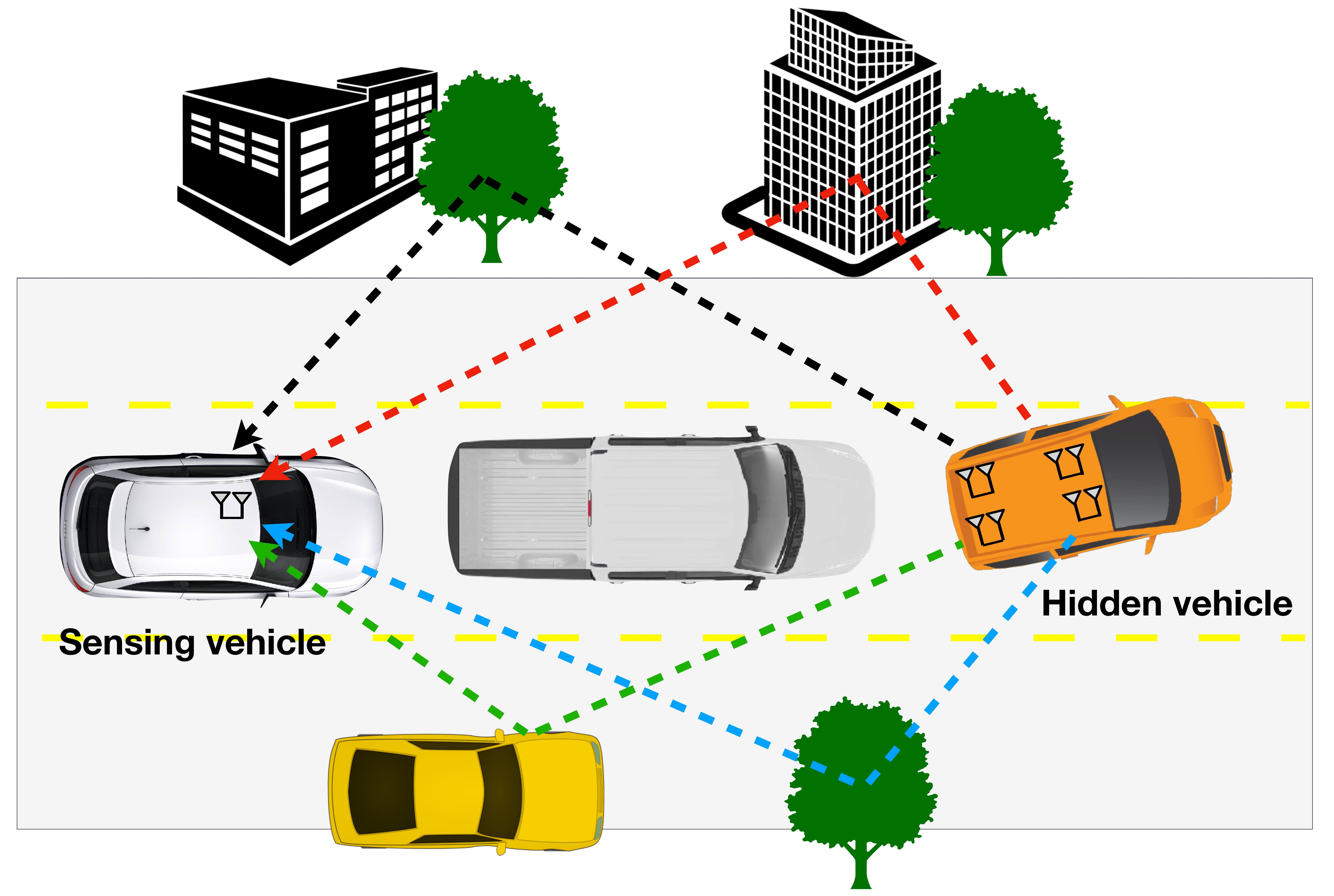}
\vspace{-10pt}
\caption{Hidden vehicle scenario with multi-path NLoS channels.}\label{v2v_network}
\vspace{-30pt}
\end{figure}

\vspace{-20pt}
\subsection{Wireless Transmission and Vehicular Positioning}\label{Sec:Intro_VehiPosi}
\vspace{-5pt}
Wireless transmission underpins different approaches for vehicular positioning  as~follows.

\subsubsection{Communication-Based Vehicular Positioning}  Nearby vehicles can position each other by  \emph{vehicle-to-everything} (V2X) communication to exchange GPS information   \cite{yao2011improving2}. The state-of-the-art protocol for connecting vehicles is the low-latency \emph{dedicated short-range communication} (DSRC) \cite{kenney2011dedicated},  a variant of the IEEE 802.11 Wi-Fi random-access protocol. The effectiveness of the protocol  has been examined in the scenario of  two-lane rural highway \cite{motro2016vehicular}. Recently, \emph{the 3rd Generation Partner Project} (3GPP) has initiated a standardization process  to realize   \emph{cellular-based V2X} (C-V2X) communication supporting  higher data  rate and larger coverage than DSRC~\cite{seo2016lte}.
To boost data rate, the technology aims at implementation in the \emph{millimetre-wave} (mmWave) spectrum where abundant of bandwidth is available  \cite{heath2016overview}.
A key challenge faced by  mmWave V2X communication is the high overhead arising from frequent  ultra-sharp  beam training and  alignment to cope with fast fading in channels between  high-speed vehicles.     
To reduce this overhead, a fast beam-alignment scheme is proposed  in \cite{perfecto2017millimeter} that leverages  matching theory and swarm intelligence to efficiently pair  vehicles. An alternative technique for accelerating beam alignment  is to leverage information generated by  either onboard RADAR \cite{gonzalez2016radar} or GPS \cite{va2018inverse} to  deduce  useful channel information.

The main challenge faced by the communication-based approach is that their reliability and latency may not meet the requirements of mission-critical scenarios such as accident avoidance at high speeds. Consider the reliability issue. The GPS information exchanged between vehicles can be highly inaccurate due to the blockage of GPS signals from satellites in urban areas or tunnels. Next, consider the latency issue. The implementation of V2X communication in a dense vehicular network incurs high communication overhead due to many practical factors including complex protocols, packet loss and retransmission, source and channel coding, limited link life time, etc. Due to the above issues, the application of the communication-based  approach is limited to  pedestrian positioning but not yet suitable    for  a latency-sensitive task like  vehicular~positioning.

\subsubsection{Reflection-Based Vehicular Positioning}
Auto-driving vehicles  are typically equipped with    RADAR and LIDAR among other sensing devices.   The two sensing technologies both adopt  the reflection-based approach, namely that the sensors detect  the reflections from vehicles and objects in the environment but using different mediums, microwaves and  laser light, respectively.

LIDAR  steers  ultra-sharp  laser beams to scan the surrounding environment and generates a dynamic high-resolution \emph{three-dimensional} (3D) map for  navigation  \cite{schwarz2010lidar}. The main drawback of LIDAR is its ineffectiveness under hostile weather conditions due to the difficulty of light in penetrating fog, snow, and  rain \cite{choi2016millimeter2}. In addition, LIDAR is currently too expensive to be practical and  the huge amount of generated data is challenging to be processed within the ultra-low latency required for safe driving.

RADAR  can locate a target   object as well as estimate its velocity  via sending a designed  waveform [e.g., \emph{frequency modulated continuous waveform} (FMCW)] and analyzing its reflection by the  object \cite{daniels2017forward}. Particularly, the metal surfaces of vehicles are  capable of reflecting microwaves with negligible absorption. For the reason,   RADAR is    popularly used  for long-range  sensing. Recent breakthroughs in mmWave RADAR  make it feasible to  deploy large-scale but high impact arrays for  sharp beamforming to achieve a high positioning accuracy \cite{menzel2012antenna}.  Compared with LIDAR,  RADAR can provide longer  sensing ranges (up to hundreds-of-meter) and retain the effectiveness  under hostile weather conditions or in an environment with poor lighting. One disadvantage of  RADAR is that it may incorrectly recognize a harmless small metal object  as a much larger object due to scattering, leading to false alarms \cite{kong2017millimeter}.

In the context of auto-driving, perhaps the most critical limitation of LIDAR, RADAR and other reflection-based technologies is  that they can  detect only vehicles with LoS since neither microwaves nor laser light can penetrate a large solid object such as a truck or a building. However, detecting HVs with \emph{non-LoS} (NLoS) is crucial for  collision avoidance  in complex scenarios such as overtaking and cross junctions as illustrated in   Fig.~\ref{v2v_network}.

\vspace{-20pt}
\subsection{Localization by Synchronous Transmission}
\vspace{-2pt}
Besides to vehicular positioning, there exists an active research area to estimate the positions of mobile devices in cellular networks, a. k. a. \emph{localization} \cite{kuutti2018survey2}, most of which rely on 
\emph{synchronous} transmissions. 
For example, a receiver (e.g., a mobile device) estimates the relative position of a transmitter [e.g., \emph{base station} (BS)] from the prior knowledge of the transmitted waveform (see e.g., \cite{kuutti2018survey2, miao2007positioning}).  
The effectiveness of this approach hinges on the key assumption of perfect {synchronization between the transmitter and receiver}. The reason is that the \emph{distance} of each propagation path can be directly converted from its  \emph{propagation delay}, i.e., ToA or TDoA (for pulse transmission)  \cite{guvenc2009survey} or \emph{frequency variation} (for FMCW transmission) \cite{feger201377}, which can be easily measured given the said synchronization. 
Localization via asynchronous transmission is possible when \emph{received-signal-strength} (RSS) information is given \cite{liu2007survey}, but the resultant accuracy is far from that via synchronous transmission. 
On the other hand, spatial parameters of individual paths, including the AoA and AoD,  can be estimated if the transmitter and receiver are equipped with antenna arrays to perform spatial filtering \cite{seow2008non}. Combining the distances and spatial parameters of multi-paths and exploiting their geometric relations enable the  localization despite the lack of LoS, given the perfect \emph{synchronization among multiple transmitters}～\cite{guvenc2009survey, wei2011aod, shikur2014tdoa}.  

The assumption of perfect transmitter-receiver or transmitter-transmitter synchronizations is reasonable for localization in cellular networks since the BSs and mobiles are relatively stationary.
However, in the auto-driving scenario, the transmitter and receiver are  a HV and a \emph{sensing vehicle} (SV), respectively,  and their synchronization is impractical especially at high speeds \footnote{The synchronization error for wireless communications is typically around $\pm0.39\mu\sec$ \cite{3GPPR36133}, which is negligible compared with the size of cyclic prefix ($4.7\mu\sec$). However, from positioning perspective, this synchronization error will result in $\pm117$m distance error, which is unacceptable and thus calls for developing a positioning technique without transmitter-receiver synchronization. }. 
 This renders  the approach based on synchronous transmission unsuitable for vehicular positioning.  
In view of the prior work, sensing a HV remains largely an open problem and tackling it is the theme of this work.

\begin{table*}
\centering
\caption{Comparison of different vehicular positioning approaches}
\vspace{-10pt}
\setlength{\tabcolsep}{2pt}
\footnotesize
\begin{tabular}{c|c|c|c|c}
\toprule
\textbf{Approach} &\textbf{Require LoS?}  & \textbf{Require Sync.?} & \textbf{Size Detection} &\textbf{Reliability} \\
\midrule
\textbf{Communication-Based} & No  & Yes & $\times$ & Low (GPS dependent)\\
\textbf{Reflection-Based} & Yes  & Yes & LIDAR: $\checkmark$ RADAR $\times$ & Low (environment dependent)\\
\textbf{Synchronous Transmission} &  No  & Yes & $\times$ & Low (need sync.) \\
\textbf{HV-Sensing (proposed)} &  No  & No  & $\checkmark$  & High \\
\bottomrule
\end{tabular}
\vspace{-30pt}
\label{Table:Compare}
\end{table*}

\vspace{-20pt}
\subsection{Main Contributions}
\vspace{-5pt}

In this work, we aim at tackling the open challenges faced by existing vehicular positioning techniques  as summarized below.

\noindent 1) \textbf{No LoS and lack of synchronization:} None of the existing techniques (see Table \ref{Table:Compare}) can be effective for sensing a HV that  has \emph{neither LoS nor synchronization} (with the SV).

\noindent 2) \textbf{Simultaneous detection of HV's position, orientation, and size:} Except for LIDAR, other existing vehicular positioning approaches are incapable of detecting the size and orientation of a HV, which are important safety information required in auto-driving.  For auto-driving safety,  all of the vehicle's safety information (position, orientation, and size) should be given simultaneously but the existing technologies cannot guarantee to achieve the goal due to various practical factors, e.g., GPS-denied environments, high mobility, and the perfect alignment among them.

\noindent 3)  \textbf{Insufficient multi-path:} Similar to the approach based on synchronous transmission, the current technology positions a HV by exploiting multi-path propagation. Such approaches  may not be feasible in the scenario with  sparse scattering and hence  insufficient paths. Tackling the challenge is important for making the technologies robust.

To tackle the above challenges, we propose a novel technology, called HV-sensing, to enable a SV to simultaneously sense the position, orientation of driving direction, and size of a HV without requiring SV-HV synchronization. The SV leverages the information of multi-path signals [including AoA, AoD and \emph{time-of-arrival} (ToA)] as well the derived  geometry relations between the paths so as to construct tractable systems of equations or optimization problems,  where the HV position, orientation of driving direction, and size are unknown variables to be found.  A set of HV-sensing techniques is designed for operation in different practical  settings ranging from low to high \emph{signal-to-noise ratios} (SNRs), single-cluster to multi-cluster HV arrays, and small to large waveform sets. The differences between the proposed HV-sensing technology  and conventional approaches are summarized in Table~\ref{Table:Compare}.

The proposed vehicular sensing technique can be used to assist and enhance the current positioning approaches by overcoming their limitations as well as enabling the vehicle to sense the state of a HV including detecting its position, orientation, and size. This technique can be integrated with current vehicular positioning system to further improving safety of auto-driving. The main contributions of this work are summarized as follows.

  {\noindent 1) \textbf{Sensing HV  position and orientation by using single cluster array:} Consider the case that the HV array contains a cluster of collocated antennas. The goal is to simultaneously estimate HV's position and orientation without SV-HV synchronization.} The HV transmits orthogonal waveforms over different antennas, enabling the SV to estimate the  multi-path information (AoA, AoD, and ToA). Given the information and when noise is negligible, a complex system of equations is constructed and solved in sequential steps at SV to obtain the desired HV-state  parameters. On the other hand, when noise is present, the sensing problem is reformulated as   \emph{least-square} (LS) estimation also solved in a sequential procedure. For the HV sensing to be feasible, the required numbers of paths are at least  4 in 2D propagation (see Proposition \ref{pro:minNumPath}) and 3 in 3D propagation (see Proposition \ref{pro:minNumPath3D}).

  {\noindent 2)  \textbf{Sensing HV  size by using multiple clusters array:} Consider the case that the HV array contains multiple clusters of  antennas that are distributed over the vehicular body. In this case, the goal is to further estimate the HV size along with its position and orientation.} Two specific schemes  are presented.  The first assumes the transmission of multiple orthogonal waveform sets so that the SV can group the paths according to their originating HV antenna clusters. Then the scheme can build on the preceding design (as shown in the point \textbf{1}) to estimate the HV size via  efficiently positioning the HV antenna clusters, which also yields the HV position and orientation. The second assumes the transmission of an identical waveform set such that the said path-grouping is infeasible. Then alternative size detection techniques are proposed based on efficient  disk or box minimization under the constraint that the disk or box encloses the HV array. The required numbers of paths for the first sensing scheme is found to be $6$ and that for the second scheme is $4$. Nevertheless, when both of two schemes are feasible, the former outperforms the latter as multiple orthogonal waveform sets help to improve sensing accuracy.

\noindent 3) \textbf{Coping with insufficient multi-path:} To make the proposed HV-sensing techniques more reliable, we further propose practical solutions for increasing the number of available paths in the case where there are insufficient for meeting the feasibility requirements of  the above HV-sensing techniques. The first solution is to \emph{combine paths} exploited in multiple time instants  and the second is to apply random directional beamforming for uncovering \emph{hidden paths} invisible in the case of isotropic HV transmission. The solutions are complementary and can be jointly implemented to maximize the number of significant paths for enhancing the sensing accuracy.

\noindent 4) \textbf{Realistic simulation:} The proposed HV-sensing techniques are evaluated using practical simulation models of highway and rural scenarios  and found to be effective.


\vspace{-10pt}
\section{System Model}\label{sec:SigModel}
\vspace{-5pt}
We consider a two-vehicle system where a SV attempts to detect the (relative) position, size, and orientation  of a HV blocked by obstacles such as trucks or buildings as illustrated in Fig.~\ref{v2v_network}.  An antenna cluster refers to a set of \emph{collocated} antennas where the half-wavelength antenna spacing is negligible compared with vehicle sizes and propagation distances. An array can comprises a single or multiple antenna clusters, referred to as a \emph{single-cluster} and a \emph{multi-cluster} array, respectively. The deployment of a single-cluster array at the HV enables  the SV to detect the HV's position and orientation. On the other hand, a multi-cluster HV array can further make it possible for the SV  to estimate the HV size. For the purpose of  exposition, in the case of multi-cluster HV array, we consider $4$ clusters located at the vertices of a rectangle representing  the vehicle. The principle of HV-sensing design in the sequel is based on the efficient detection of the clusters' positions and thus  can be  straightforwardly extended  to other clusters' topologies with irregular clusters' distributions. For HV-sensing in both scenarios, the SV requires only a single-cluster array. Signal propagation is assumed to be contained in the 2D plane and the  results are subsequently extended  to the 3D propagation. The channel model, V2V transmission, and sensing problem are described in the following sub-sections. The notations used frequently throughout the paper are summarized in Table \ref{Table:Notation}.

\begin{table}[t]
\centering
\caption{Summary of Notation}
\vspace{-10pt}
\setlength{\tabcolsep}{2pt}
\footnotesize
\begin{tabular}{c|c}
\toprule
\textbf{Notation} &\textbf{Meaning} \\
\midrule
AoA, AoD, ToA, orientation of driving direction & $\theta$, $\varphi$, $\lambda$, $\omega$  \\
Propagation distance, propagation distance after reflection & $d$, $\nu$  \\
Propagation delay of signal,  synchronization gap between HV and SV & $\tau$, $\Gamma$ \\
Number of transmit, receive antennas per cluster & $M_t$, $M_r$ \\
Carrier frequency and bandwidth of signal & $f_c$, $B_s$ \\
Transmitted, received signal (waveform) vector & $\mathbf{s}(t)$, $\mathbf{r}(t)$ \\
Response vector of transmit, receive antennas per cluster & $\mathbf{a}(\varphi)$, $\mathbf{b}(\theta)$ \\
Channel coherence time, signal waveform duration & $T_c$, $T_w$ \\
Number of observed paths and its index & $P$, $p$ \\
Azimuth and elevation angles of AoA, AoD, orientation in 3D & $\{\theta, \vartheta\}$, $\{\varphi, \psi\}$, $\{\omega, \varrho\}$\\
Length and width of vehicle & $L$, $W$ \\
Number of repetitive signal transmissions and its interval & $Q$, $\Delta$\\
\bottomrule
\end{tabular}
\vspace{-30pt}
\label{Table:Notation}
\end{table}

\vspace{-20pt}
\subsection{Multi-Path NLoS Channel}\label{subsec:ScenarioDescription}
\vspace{-5pt}

The NLoS channel between SV and HV contains multi-paths reflected by a set of scatterers. Consider the characteristics of V2V channel \cite{karedal2009geometry}. We make the following assumption.
\vspace{-10pt}
\begin{assumption}[Single-Bounce Scattering]\label{assump:SingleBounce}\emph{The single-bounce scattering is used to model the V2V propagations that the NLoS signals are assumed to have only one reflection due to~scatterers.}
\end{assumption}
\vspace{-10pt}

 {Assumption \ref{assump:SingleBounce} is widely used in the literature of localization via NLoS paths \cite{miao2007positioning, seow2008non}.
Various methods have been proposed to detect single-bounce paths among multiple-bounce ones such a proximity detection \cite{seow2008non}, and a joint ToA-and-signal strength based detection \cite{yu2009statistical}, making sense to use Assumption \ref{assump:SingleBounce} in practice.
 Based on Assumption \ref{assump:SingleBounce}, a 2D Cartesian coordinate system is considered as illustrated in Fig.~\ref{2DSignalModel}  where the  SV array  is located at the origin and  the $X$-axis is aligned with the orientation of SV. Further consider a typical antenna cluster  at the HV.  Each NLoS signal path from the HV antenna cluster  to the SV array can be characterized by the following five parameters: the AoA at the SV denoted by $\theta$; the AoD at the HV denoted by $\varphi$; the orientation of the HV denoted by $\omega$; and the propagation distance denoted by $d$ which is divided into   the propagation distance after reflection, denoted by $\nu$, and the remaining distance $(d - \nu)$. The AoD and AoA are defined as azimuth angles relative to orientations of HV and SV, respectively. Fig.~\ref{2DSignalModel} graphically shows the definitions of the above parameters. 

\vspace{-20pt}
\subsection{Hidden Vehicle Transmission}
\vspace{-5pt}
To enable sensing at the SV, the HV transmits a set of waveforms defined as follows. Each antenna cluster at HV  has $M_t$ antennas with at least half-wavelength spacing between adjacent antennas. Consider a typical antenna cluster. A set of   orthogonal waveforms are  transmitted over different antennas\footnote{It is implicitly assumed that enough orthogonal waveforms are given for HV's multi-antenna transmission a.k.a. \emph{orthogonal multiple access} (OMA). It is interesting to extend the current design to \emph{non-orthogonal multiple access} (NOMA) \cite{Ding2017} against case of insufficient orthogonal waveforms, which is outside the scope of current work. }. Let $s_m(t)$ be the finite-duration  baseband waveform in $t\in [0, T_s]$ assigned to the $m$-th HV antenna with the bandwidth $B_s$. Then the waveform orthogonality is specified by  $\int_0^{T_s} s_{m_1}(t) s^*_{m_2}(t)dt=\delta(m_1-m_2)$ with the delta  function  $\delta(x)=1$ if $x=0$ and $0$ otherwise. The transmitted waveform vector  for the $k$-th  HV antenna cluster is  $\mathbf{s}^{(k)}(t) = [s^{(k)}_1(t), \cdots, s^{(k)}_{M_t}(t)]^{\mathrm{T}}$. In the case of multi-cluster HV array, the waveform sets for different clusters are either \emph{identical} or \emph{orthogonal} with each other. The use of orthogonal waveform sets allows SV to group the detected paths according to their originating antenna clusters as elaborated in the sequel, and hence this case is referred to as \emph{decoupled clusters}. Then the other case is called \emph{coupled clusters}. With the prior knowledge of transmitted waveforms, the SV with $M_r$ antennas can scan   the received signal due to the HV transmission to resolve multi-path as discussed in the next sub-section.

\begin{figure}[t]
\centering
\includegraphics[width=7.5cm]{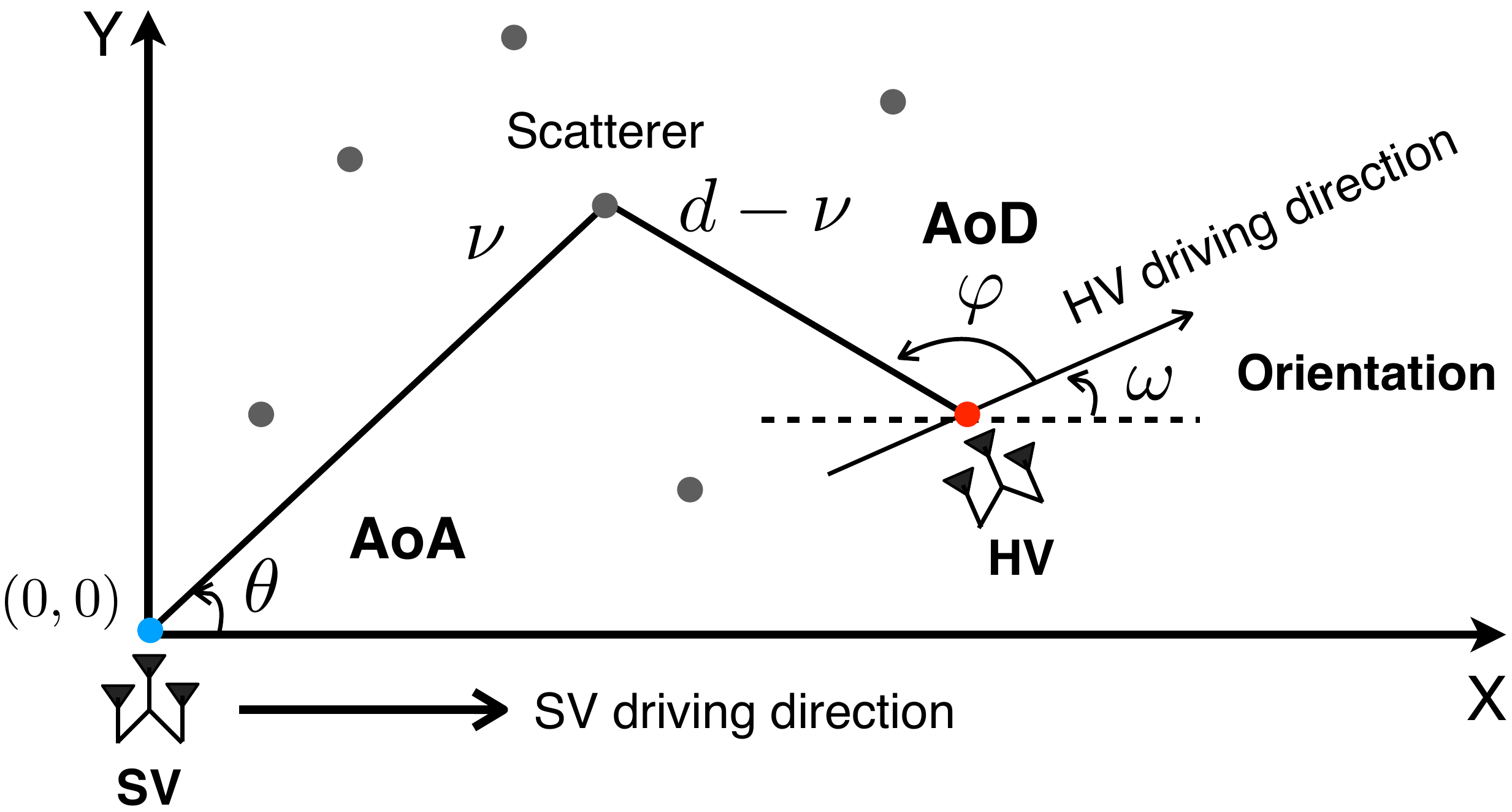}
\vspace{-10pt}
\caption{The geometry of a 2D propagation path and the definitions of  parameters.}\label{2DSignalModel}
\vspace{-30pt}
\end{figure}

The expression of the received signal is obtained as follows.  Consider a typical HV antenna cluster.  Based on the far-field propagation model, the cluster response vector is  a function of the AoD $\varphi$ defined as
\vspace{-15pt}
\begin{align}\label{HVsteeringVector}
\mathbf{a}(\varphi) = [ \exp(j 2\pi f_c \alpha_{1}(\varphi)), \cdots,\exp(j 2\pi f_c \alpha_{M_t}(\varphi))   ]^{\mathrm{T}},
\end{align}
where $f_c$ denotes the carrier frequency and $\alpha_m(\varphi)$ is the propagation time difference between received signals departing from the HV's  $m$-th antenna and the first antenna, i.e., $\alpha_1(\varphi)=0$.
The response  vector of the SV array is written  in terms of AoA  $\theta$ as
\vspace{-10pt}
\begin{align}\label{SVsteeringVector}
\mathbf{b}(\theta) = [ \exp(j 2\pi f_c \beta_{1}(\theta)), \cdots,\exp(j 2\pi f_c \beta_{M_r}(\theta))   ]^{\mathrm{T}},
\end{align}
where $\beta_m(\theta)$ refers to the propagation time difference between received signals arriving at the SV's $m$-th antenna and the first antenna. We assume that SV has prior knowledge of both the HV-and-SV array configurations and thereby the response functions $\mathbf{a}(\varphi)$ and $\mathbf{b}(\theta)$. This is feasible by standardizing  the vehicular arrays' topology.   In addition, the Doppler effect is ignored based on the assumption that the Doppler frequency shift is much smaller than the waveform bandwidth and thus does not affect waveform orthogonality. Let $k$ denote the  index of HV antenna cluster and $P^{(k)}$ denote the number of received paths originating  from the $k$-th HV antenna cluster. The total number of  paths arriving  at SV is  $P = \sum_{k=1}^K P^{(k)}$. We represent  the received signal vector at SV as $\mathbf{r}(t) = [r_1(t), \cdots, r_{M_r}(t)]^{\mathrm{T}}$ that can be written in  terms of $\mathbf{s}(t)$, $\mathbf{a}(\varphi)$ and $\mathbf{b}(\theta)$  as
\vspace{-5pt}
\begin{align}
\mathbf{r}(t) = \sum_{k=1}^{K} \sum_{p=1}^{P^{(k)}} \gamma_{p}^{(k)} \mathbf{b}\l(\theta_{p}^{(k)}\r) \mathbf{a}^{\mathrm{T}}\l(\varphi_{p}^{(k)}\r)\mathbf{s}\l(t - \lambda_p^{(k)} \r) + \mathbf{n}(t), 
\end{align}
where $\gamma_{p}^{(k)}$ and  ${\lambda_p^{(k)}}$ denote the complex channel coefficient and
 ToA of path $p$ originating  from the  $k$-th HV antenna cluster, respectively,
and  $\mathbf{n}(t)$ represents channel noise. With \emph{no synchronization} between  the HV and SV,  the SV has no information of HV's transmission timing. Therefore, it is important to note that $\lambda_{p}^{(k)}$ differs from the corresponding propagation delay, denoted by ${\tau_p^{(k)}}$, with ${\tau_p^{(k)}}=\frac{d_p^{(k)}}{c}$ and $d_p^{(k)}$ being the propagation distance. Given an unknown clock-synchronization gap between the HV and SV denoted as  $\Gamma$,  $\tau_{p}^{(k)}=\lambda_{p}^{(k)}-\Gamma$.

\begin{figure}[t]
\centering
   \includegraphics[width=12cm]{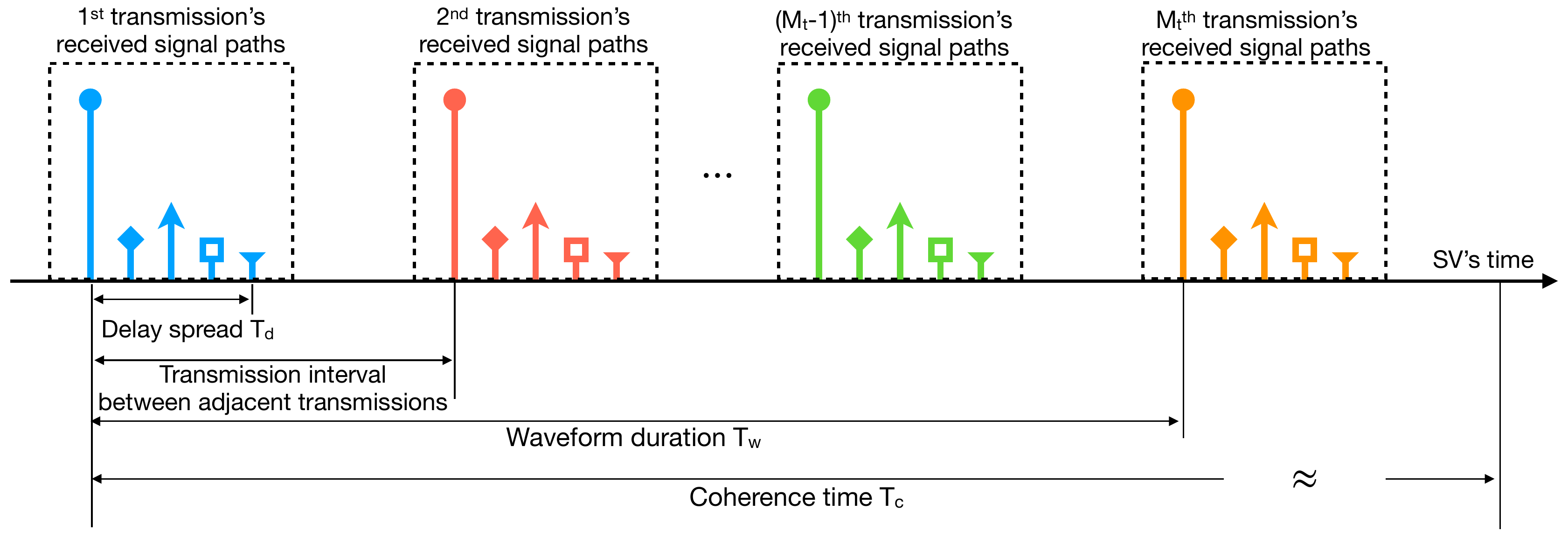}
\vspace{-10pt}
\caption{The relation among delay spread, waveform duration, and coherence time. A Fourier matrix is selected as a reference signal of the waveform of which the length is the number of transmit antenna $M_t$.}\label{fig:Waveform}
\vspace{-30pt}
\end{figure}

\vspace{-10pt}
\begin{remark}[Extension to Multi-Vehicle Auto-Driving Scenario]\label{remark:MultiVehicle}
\emph{The current two-vehicle auto-driving scenario can be extended to multi-vehicle scenario straightforwardly by jointly using the following two approaches. 
\textbf{1) Waveform pool \& sensing}: Interference-free multi-vehicle sensing is possible if multiple orthogonal waveforms are given, called a \emph{waveform pool}. The list of waveforms in the pool is periodically broadcast by a \emph{road side unit} (RSU), and each vehicle knows the list. Before selecting a waveform, the usages of all waveforms in the list are sensed. It is analogy to a carrier sensing mechanism in WiFi communication, and it is thus called a \emph{waveform sensing}. To avoid waveform collision, one waveform is selected in the list except the ones being used by nearby vehicles.  
\noindent \textbf{2) Geo-zoning}: In spite of using waveform sensing, more than two vehicles could select the same waveform simultaneously if they are out of their sensing coverage, named as  \emph{hidden vehicle problem} (HVP).  To avoid HVP, a spatial isolation should be created where vehicles in different locations would be limited to select the resources for transmission from a certain time-frequency set, based on their absolute geographical location. In 3GPP specification, it is called geo-zoning or zone-based resource allocation \cite{3GPPR1}.}
\end{remark}
\vspace{-20pt}

\begin{remark}[Waveform Duration vs. Channel Coherent Time]\label{remark:waveform}
\emph{It is essential to study the condition on channel dynamics for accurate and stable positioning.
Specifically, the assumption of invariant ToA, AoA and AoD holds if the V2V channel is largely unchanged in the waveform duration:
\vspace{-10pt}
\begin{align}\label{Eq:SlowFadingCond}
T_c\gg T_w,
\end{align}
where $T_c$ and $T_w$ denote the channel coherence time and waveform duration, respectively. To verify the condition, we analyze the two metrics as follows.
\begin{itemize}
\item \emph{Coherence time:} It is well-known that $T_c$ is inversely proportional to the maximum Doppler frequency $f_\Delta$ as
$T_c\approx \frac{1}{f_\Delta}=\frac{c}{f_c v}$
with the relative velocity of HV to SV $v$ (HV and SV are driving in in opposite directions) and the carrier frequency $f_c$. When $v=200$ km/h and $f_c=5.9$ GHz,  $T_c$ approximately becomes $0.9153$ ms.
\item \emph{Waveform duration:} We select a $M_t$ by $M_t$ Fourier matrix as a codeword of the orthogonal waveform  \cite{rossi2014spatial}.
To guarantee the waveform's orthogonality, the interval between adjacent code transmissions should be larger than the channel's delay-spread, denoted by $T_d$, enabling to avoid inter-code interference (see Fig. \ref{fig:Waveform}). In other words, we should satisfy
$T_w
\geq M_t T_d$,
where $M_t$ refers to the length of the codeword.
Note the coverage of V2V channel with $f_c=5.9$ GHz is less than $1$ km \cite{Flament2018}, which is
translated to the maximum delay spread $\frac{10^3}{c}\approx 3.33$ $\mu$s.
Given $M_t=20$, we obtain $T_w\geq 0.0667$ ms.
\end{itemize}
As a result, the condition in \eqref{Eq:SlowFadingCond} is satisfied. It is worth mentioning that the technique proposed in the sequel is based on this
\emph{one-way} transmission of which the entire duration is significantly shorter than
conventional two-way positioning approaches (e.g., \emph{round-trip-time} (RTT) based), enabling accurate and stable positioning  in highly dynamic~environments.}
\end{remark}

\vspace{-25pt}
\subsection{Estimations of AoA, AoD, and ToA}\label{Subsec:EstAoAAoDToA}
The sensing techniques in the sequel assume that the SV has the  knowledge of AoA, AoD, and ToA of each receive NLoS signal path, say path $p$, {denoted by $\l\{\theta_{p}, \varphi_{p}, \lambda_{p}\r\}$ where $p \in \mathcal{P} = \{1, 2, \cdots, P\}$}. The knowledge can be acquired by applying  classical parametric estimation techniques briefly described  as follows.

 \noindent \textbf{1) Sampling}: The received analog signal $\mathbf{r}(t)$ and the waveform vector $\mathbf{s}(t)$ are sampled at the Nyquist rate $B_s$ and converted to the digital signal vectors $\mathbf{r}[n]$ and $\mathbf{s}\l[n\r]$, respectively.

 \noindent \textbf{2) Matched Filter}: The sequence of $\mathbf{r}[n]$ is matched-filtered by $\mathbf{s}\l[n\r]$. The resultant $M_r\times M_t$ coefficient matrix $\mathbf{Y}[z]$ is given by $\mathbf{Y}[z]=\sum_n \mathbf{r}[n]\mathbf{s}^H\l[n - z\r]$.
    The sequence of ToAs $\{\lambda_{p}\}$ can be estimated by detecting peak points of the matrix norm  $\|\mathbf{Y}[z]\|$, denoted by $\{z_p\}$, which can be converted into time by multiplication with  the time resolution $\frac{1}{B_s}$. One peak point can contain multiple signal paths if the signals arrive  within the same  sampling interval. 

\noindent \textbf{3) Multi-Path Estimation}: Given $\{\mathbf{Y}[z_p]\}$, AoAs and AoDs are jointly estimated by scanning two-dimensions angular domains by leveraging the prior knowledge of $\mathbf{a}(\varphi)$ in \eqref{HVsteeringVector} and $\mathbf{b}(\theta)$ in \eqref{SVsteeringVector}. It is called a 2D-\emph{multiple signal classification} (MUSIC) algorithm, which is the most widely used subspace-based detection method. The estimated AoA $\{\theta_p\}$ and AoD $\{\varphi_p\}$ are paired with the corresponding estimated ToA $\{\lambda_p\}$, which jointly characterize path $p$.

It is worth mentioning that only a small portion of signal paths can be observed among the entire ones because some of them experience severe signal attenuations due to path-loss and small-scale fading and their received signal strengths are not enough to estimate AoD/AoA accurately.

\vspace{-15pt}
\subsection{Hidden Vehicle Sensing Problem}\label{Subsec:TechnicalChallenging}
\vspace{-5pt}

The SV attempts to sense the HV's position, size, and orientation, which can be obtained by using parameters of AoA $\theta$, AoD $\varphi$, orientation $\omega$, distances $d$ and $\nu$, and location of multi-cluster HV array \footnote{The proposed sensing technique is to exploit the geometry relation between a few channel parameters of a signal path i.e., TDoA/AoA/AoD, which is analogy to transmitting a reference signal for conventional wireless channel estimation (e.g., \cite{ozdemir2007channel}).} Noting the first two parameters are obtained based on the estimations in Section~\ref{Subsec:EstAoAAoDToA} and the goal is to estimate the remaining parameters.

\vspace{-10pt}
\section{Sensing Hidden Vehicle with a Single-Cluster Array}\label{sec:singleCluster}
\vspace{-5pt}
In this section, we consider the scenario that a single-cluster array is deployed at HV. Then this section focuses on designing the sensing techniques for SV to detect 1) the HV position (i.e., position of the single-cluster array), specified by the coordinate $\bp = (x, y)$,  and 2) the HV orientation, specified by $\omega$ (see Fig.~\ref{2DSignalModel}). Based on the multi-path-geometry, $\bp$ is described as
\vspace{-5pt}
\begin{align}\label{TXposition2D}
\begin{cases}
x = \nu_{p} \cos(\theta_{p}) - (d_{p} - \nu_{p}) \cos(\varphi_{p} + \omega),  \\
y = \nu_{p} \sin(\theta_{p}) - (d_{p} - \nu_{p}) \sin(\varphi_{p} + \omega),
\end{cases}
~~\forall p \in \mathcal{P}.
\end{align}
The prior knowledge that  SV has for sensing is the parameters of $P$ NLoS paths estimated as described in Section \ref{Subsec:EstAoAAoDToA}. Each path, say path $p$,  is determined  by the parametric set $\{\theta_p, \phi_p, \lambda_p\}$ and orientation $\omega$ as \eqref{TXposition2D} shows. Then given the equations in \eqref{Eq:ProbForm1}, the \textbf{sensing problem} for the current scenario reduces to
\vspace{-10pt}
\begin{equation}\label{Eq:ProbForm1}
\bigcup_{p \in \mathcal{P}} \{\theta_p, \phi_p, \lambda_p\} \Rightarrow \{\bp, \omega\}.
\end{equation}

\vspace{-20pt}
\subsection{Sensing Feasibility Condition}
In this subsection, it is shown that for the HV-sensing to be feasible, there should exist at least \emph{four} NLoS paths. To this end, by using \eqref{TXposition2D} and multi-path-geometry, we can obtain the following system of equations:
\vspace{-10pt}
\begin{align}\label{FGrelation}\tag{E1}\!\!\!\!\!\!
\begin{cases}
x_{p} = \nu_{p} \cos(\theta_{p}) - (d_{p} - \nu_{p}) \cos(\varphi_{p} + \omega) = \nu_{1} \cos(\theta_{1}) - (d_{1} - \nu_{1}) \cos(\varphi_{1} + \omega),  \\
y_{p} = \nu_{p} \sin(\theta_{p}) - (d_{p} - \nu_{p}) \sin(\varphi_{p} + \omega)  =\nu_{1} \sin(\theta_{1}) - (d_{1} - \nu_{1}) \sin(\varphi_{1} + \omega),
\end{cases}
 p  \in \mathcal{P},
\end{align}
where $(x_{p}, y_{p})$ denotes the coordinate characterized via  path $p$.  The number of  equations in \eqref{FGrelation} is $2(P-1)$, and the above system of equations has a unique solution when the dimensions of unknown variables are less than $2(P-1)$.
Since the   AoAs $\{\theta_{p}\}$ and AoDs $\{\varphi_{p}\}$ are known,  the number of unknowns is  $(2P+1)$ including the propagation distances $\{d_{p}\}$, $\{\nu_{p}\}$, and orientation $\omega$. To further reduce the number  of unknowns, we use the propagation time difference between signal paths also known as TDoAs, denoted by $\{\rho_{p}\}$, which can be obtained from the difference of ToAs as
$\rho_{p}=\lambda_{p}-\lambda_{1}$
where   $\rho_{1}=0$.
The propagation distance of signal path $p$, say $d_{p}$, is then expressed in terms of $d_{1}$ and $\rho_{p}$ as
\vspace{-10pt}
\begin{align}\label{Eq:TDoAandPropagationDist}
d_{p}=c(\lambda_{p}-\Gamma)
=c(\lambda_1-\Gamma)+c(\lambda_{p}-\lambda_{1})=d_{1}+c\rho_{p},
\end{align}
for  $p \in \{2, \cdots, P\}$. Substituting the above $(P-1)$ equations into   \eqref{FGrelation} eliminates the unknowns $\{d_2, \cdots, d_P\}$ and hence reduces  the number  of unknowns from $(2P+1)$ to $(P+2)$. As a result, \eqref{FGrelation} has a unique solution when $2(P-1)\geq P+2$.

\vspace{-10pt}
\begin{proposition}[Sensing Feasibility Condition]\label{pro:minNumPath}
\emph{To sense the position and orientation of a HV equipped with a single-cluster array, at least four NLoS signal paths are required: $P\geq 4$.}
\end{proposition}
\vspace{-10pt}

\vspace{-10pt}
\begin{remark}[Asynchronization]\emph{Recall that one sensing challenge is asynchronization between HV and SV represented by  $\Gamma$, which is a latent variable we cannot observe explicitly. In our proposed approach, based on the fact that all NLoS paths experience the same synchronization gap, we do not need to consider $\Gamma$ in the positioning procedure.}
\end{remark}

\vspace{-25pt}
\subsection{Hidden Vehicle Sensing  without Noise}\label{Sec:PositionSingleMIMO}
\vspace{-10pt}

Consider the case of a high receive \emph{signal-to-noise ratio} (SNR) where noise can be neglected, i.e., the estimations of AoA/AoD/ToA $\{\theta_p, \phi_p, \lambda_p\}$ are perfect. Then the HV-sensing problem in \eqref{Eq:ProbForm1} is translated to solve the system of equations in  \eqref{FGrelation}. One challenge is that the unknown orientation $\omega$ introduces  nonlinear relations, namely  $\cos(\varphi_p + \omega)$ and $\sin(\varphi_p + \omega)$, in the equations. To overcome the difficulty, we adopt the following two-step approach: 1) estimate the correct orientation $\omega^*$ via its discriminant introduced in the sequel; 2) given $\omega^*$, the equation becomes linear and thus can be solved via  LS estimator, giving the position $\bp^*$.
To this end, the equations in  \eqref{FGrelation} can be  arranged in a matrix form as
\vspace{-10pt}
\begin{align}\label{ax=b}\tag{E2}
\mathbf{A}(\omega) \mathbf{z}= \mathbf{B}(\omega),
\end{align}
where $\mathbf{z}=(\mathbf{v}, d_1)^{\mathrm{T}} \in \mathds{R}^{(P+1)\times 1}$ and  $\mathbf{v} = \{\nu_1, \cdots,   \nu_P\}$. For matrix $\mathbf{A}(\omega)$, we have
\vspace{-5pt}
\begin{align}\label{matrixA}
\mathbf{A}(\omega) &=
\begin{bmatrix}
  \mathbf{A}^{(\cos)}(\omega) \\ \mathbf{A}^{(\sin)}(\omega)
\end{bmatrix}
\in \mathds{R}^{2(P-1) \times (P+1)},
\end{align}
where $\mathbf{A}^{(\cos)}(\omega)$ is
\begin{align}\label{A}
\begin{bmatrix}
    a_{1}^{(\cos)} & -a_{2}^{(\cos)}  & 0 & \cdots &0 & a_{1,2}^{(\cos)}  \\
    a_{1}^{(\cos)}   & 0 & -a_{3}^{(\cos)}  & \cdots& 0 & a_{1,3}^{(\cos)} \\
    \vdots  & \vdots & \vdots & \ddots & \vdots & \vdots \\
    a_{1}^{(\cos)}   & 0 & 0 &\cdots &-a_{P}^{(\cos)}  & a_{1,P}^{(\cos)} \\
\end{bmatrix}
\end{align}
with  $a_{p}^{(\cos)}  = \cos(\theta_p) + \cos(\varphi_p + \omega)$ and $a_{1,p}^{(\cos)} =  \cos(\varphi_p +\omega) - \cos(\varphi_1 + \omega)$, and $\mathbf{A}^{(\sin)}(\omega)$  is obtained by replacing all $\cos$ operations in \eqref{A} with $\sin$ operations. Next,
\vspace{-5pt}
\begin{align}\label{matrixB}
\mathbf{B}(\omega) =
\begin{bmatrix}
  \mathbf{B}^{(\cos)}(\omega) \\ \mathbf{B}^{(\sin)}(\omega)
\end{bmatrix}
\in \mathds{R}^{2(P-1)\times 1},
\end{align}
where $\mathbf{B}^{(\cos)}(\omega) =[
    c \rho_2 \cos(\varphi_2 + \omega), c \rho_3 \cos(\varphi_3 + \omega), \cdots, c \rho_P  \cos(\varphi_P + \omega)
]^{\mathrm{T}}$, and $\mathbf{B}^{(\sin)}(\omega)$ is obtained by replacing all $\cos$ in $\mathbf{B}^{(\cos)}(\omega)$ with $\sin$.

\noindent \underline{1) Computing $\omega^*$:} Note that \ref{ax=b} becomes an \emph{over-determined} linear system of equations if $P \geq 4$ (see Proposition \ref{pro:minNumPath}), providing the following discriminant of orientation $\omega$. Since the equations in \eqref{matrixA} are based on the geometry of multi-path propagation and HV orientation as illustrated in Fig.~\ref{2DSignalModel}, there exists a unique solution for  the equations. Then we can obtain from  \eqref{matrixA} the following  result, which is useful for computing $\omega^*$.

\vspace{-15pt}
\begin{proposition}[Discriminant of Orientation]\label{pro:discriminant}
\emph{With $P \geq 4$, the unique $\omega^*$ exists when $\mathbf{B}(\omega^*)$ is orthogonal to the null column space of $\mathbf{A}(\omega^*)$ denoted by $\mathsf{null}(\mathbf{A}(\omega^*)^{\mathrm{T}})\in \mathds{R}^{2(P-1)\times (P-3)}$:
\vspace{-10pt}
\begin{align}\label{Eq:OrthogonalNullSpace}
\mathsf{null}(\mathbf{A}(\omega^*)^{\mathrm{T}})^{\mathrm{T}}\mathbf{B}(\omega^*)=\mathbf{0}.
\end{align}}
\end{proposition}
\vspace{-10pt}

Given  this discriminant, a simple 1D search can be performed over  $[0, 2\pi]$ to find  $\omega^*$.

\noindent \underline{2) Computing $\bp^*$:}
Given the $\omega^*$, (\ref{ax=b}) can be  solved by
\vspace{-10pt}
\begin{align}\label{LSresult}
\mathbf{z}^* = \l[ \mathbf{A}(\omega^*)^{\mathrm{T}} \mathbf{A}(\omega^*)  \r]^{-1} \mathbf{A}(\omega^*)^{\mathrm{T}} \mathbf{B}(\omega^*).
\end{align}
Then  the HV position   $\bp^{*}$ can be computed by substituting \eqref{Eq:OrthogonalNullSpace} and \eqref{LSresult} into \eqref{Eq:TDoAandPropagationDist} and \eqref{FGrelation}.


\vspace{-15pt}
\subsection{Hidden Vehicle Sensing  with Noise}\label{subsection:NoiseHVSensing}
\vspace{-5pt}
In the presence of significant channel noise,  the estimated  AoAs/AoDs/ToAs contain errors. Consequently, HV-sensing is based on  the noisy versions of matrix $\mathbf{A}(\omega)$ and $\mathbf{B}(\omega)$, denoted by $\tilde{\mathbf{A}}(\omega)$ and $\tilde{\mathbf{B}}(\omega)$, which do not satisfy the  equations  in (\ref{ax=b}) and \eqref{Eq:OrthogonalNullSpace}. To overcome the difficulty, we develop a sensing technique by converting the equations into   minimization problems whose solutions are robust against noise.

\noindent \underline{1) Computing $\omega^*$:} Based on \eqref{Eq:OrthogonalNullSpace}, we formulate  the following problem to find  the orientation $\omega$:
\vspace{-25pt}
\begin{align}\label{NullSearching_Noisy}
\omega^* = \arg\min_{\omega} \l[\mathsf{null}(\tilde{\mathbf{A}}(\omega)^{\mathrm{T}})^{\mathrm{T}}\tilde{\mathbf{B}}(\omega)\r],
\end{align}
Solving the problem relies on a 1D search over $[0, 2\pi]$.

\noindent \underline{2) Computing $\bp^*$:} Next, given $\omega^*$, the optimal $\mathbf{z}^*$ can be derived by
using the  LS estimator that minimizes the squared Euclidean distance as
\vspace{-10pt}
\begin{align}
\mathbf{z}^*  &= \arg \min_{\mathbf{z}} \| \tilde{\mathbf{A}}(\omega^*) \mathbf{z} - \tilde{\mathbf{B}}(\omega^*) \|^2 =\l[ \tilde{\mathbf{A}}(\omega^*)^{\mathrm{T}} \tilde{\mathbf{A}}(\omega^*)  \r]^{-1} \tilde{\mathbf{A}}(\omega^*)^{\mathrm{T}} \tilde{\mathbf{B}}(\omega^*), \label{LSEstimator_Noisy}
\end{align}
where $\|\cdot \|$ is a Euclidean norm. It is shown that \eqref{LSEstimator_Noisy} has the same structure as \eqref{LSresult}. Last, the origins of all paths $\{(x_p, y_p)\}$ can be  computed  using the parameters $\{\mathbf{z}^* , \omega^*\}$ as illustrated in \eqref{FGrelation}. Averaging these origins gives the estimate of the HV position $\bp^*=(x^*, y^*)$ with $x^* = \frac{1}{P}\sum_{p=1}^{P} x_p$ and $y^* = \frac{1}{P} \sum_{p=1}^{P} y_p$.
\vspace{-10pt}
\begin{remark}[Feasibility to LoS Vehicle Sensing]\label{remark:LOS}
\emph{The proposed technique is also feasible to sense a vehicle in LoS since the LoS path is a special case of NLoS path considering that one virtual scatterer is located on the LoS path. The accuracy of LoS vehicle sensing is much better than that of HV sensing because the received signal power of LoS is much larger that of NLoS, providing more accurate AoA/AoD/ToA estimations. The detailed comparison is given in  Section \ref{sec:simulation}.}
\end{remark}
\vspace{-15pt}

\begin{figure}[t]
\centering
\includegraphics[width=7.5cm]{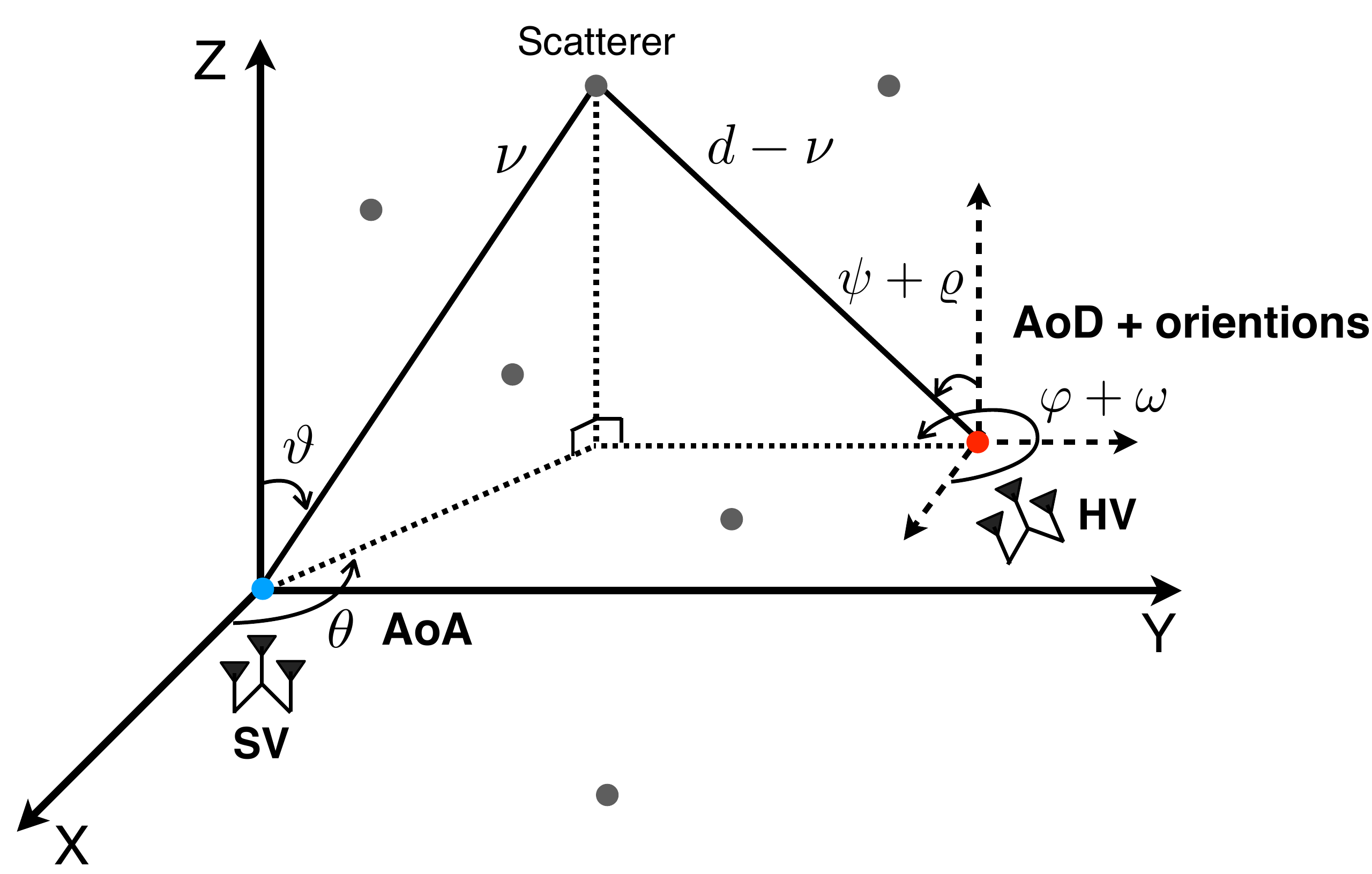}
\vspace{-20pt}
\caption{3D propagation model.}\label{3DSignalModel}
\vspace{-30pt}
\end{figure}

\vspace{-20pt}
\subsection{Extension to 3D Propagation}\label{Sec:3DSignalModel}
\vspace{-7pt}

Consider the scenario that propagation paths lie in the 3D Euclidean space instead of the 2D plane previously assumed.    As shown in Fig.~\ref{3DSignalModel}, the main differences  are the elevation angles added to the  AoAs, AoDs, and HV orientation. Specially, the AoA includes two angles: $\theta$ (azimuth) and $\vartheta$ (elevation) and AoD consists $\varphi$ (azimuth) and $\psi$ (elevation).  The estimations of AoAs and AoDs in the 3D  model can be jointly estimated via various approaches, e.g., MUSIC algorithm for 3D signal detection (see e.g., \cite{ping20093d}). The HV orientation also includes two unknowns: $\omega$ (azimuth) and $\varrho$ (elevation). The coordinates of HV, denoted by ${\mathbf{p}} = ({x}, {y}, {z})^{\mathrm{T}}$, are given as
\vspace{-10pt}
\begin{align}
\begin{cases}
{x} = \nu_{p} \sin(\vartheta_{p}) \cos(\theta_{p}) - (d_{p} - \nu_{p}) \sin(\psi_{p}+ \varrho) \cos(\varphi_{p} + \omega),  \\
{y} = \nu_{p} \sin(\vartheta_{p}) \sin(\theta_{p}) - (d_{p} - \nu_{p}) \sin(\psi_{p}+\varrho) \sin(\varphi_{p} + \omega), \\
{z} = \nu_{p} \cos(\vartheta_{p}) - (d_{p} - \nu_{p}) \cos(\psi_{p}+\varrho),
\end{cases}
\end{align}
where $\forall p \in \mathcal{P}$.

 Similar to \eqref{FGrelation}, the following system of equations is constructed for 3D propagation:
 \vspace{-10pt}
\begin{align}\label{3DFGrelation}\tag{E3}
\begin{cases}
\nu_{p} \sin(\vartheta_{p}) \cos(\theta_{p}) - (d_{p} - \nu_{p}) \sin(\psi_{p}+ \varrho) \cos(\varphi_{p} + \omega)\\
= \nu_{1} \sin(\vartheta_{1}) \cos(\theta_{1}) - (d_{1} - \nu_{1}) \sin(\psi_{1}+ \varrho) \cos(\varphi_{1} + \omega),  \\
\nu_{p} \sin(\vartheta_{p}) \sin(\theta_{p}) - (d_{p} - \nu_{p}) \sin(\psi_{p}+\varrho) \sin(\varphi_{p} + \omega)\\
= \nu_{1} \sin(\vartheta_{1}) \sin(\theta_{1}) - (d_{1} - \nu_{1}) \sin(\psi_{1}+\varrho) \sin(\varphi_{1} + \omega), \\
\nu_{p} \cos(\vartheta_{p}) - (d_{p} - \nu_{p}) \cos(\psi_{p}+\varrho)
= \nu_{1} \cos(\vartheta_{1}) - (d_{1} - \nu_{1}) \cos(\psi_{1}+\varrho),
\end{cases}
\end{align}
where $\forall p \in \mathcal{P}$.

It is shown that the number of equations and the number of unknown variables are $3(P-1)$ and $(P+3)$, respectively.
For  the HV-sensing  problem to be  solvable, we require $3(P-1) \geq P+3$, which leads to  the following proposition.
\vspace{-10pt}
\begin{proposition}[Sensing Feasibility Condition for 3D]\label{pro:minNumPath3D}
\emph{Consider the  3D propagation model. To sense the position and orientation of HV provisioned with  a single-cluster array, at least three NLoS signal paths are required, i.e., $P\geq 3$.}
\end{proposition}
\vspace{-10pt}

Compared with 2D propagation, the minimal number of required signal paths is reduced  because extra information can be extracted from one  additional dimension (i.e., elevation angles information of AoAs, AoDs) of each signal path. A similar methodology described in Sections \ref{Sec:PositionSingleMIMO} and \ref{subsection:NoiseHVSensing} can be easily modified for 3D propagation by applying  a 2D search based discriminant to find $\omega$ and $\varrho$ over $[0, 2\pi]$ and $[0, \pi]$, respectively.

\vspace{-10pt}
\section{Sensing Hidden Vehicle  with a Multi-Cluster Array}\label{Sec:Multi-Cluster}
\vspace{-10pt}
The preceding section targets the scenario that the HV is provisioned with  a single-cluster array, allowing the SV to sense the HV position and orientation. In this section, we consider the scenario where a multi-cluster array is deployed at HV so that SV can sense HV's array size (approximating the HV size) in addition to its  position and orientation. Sensing techniques are designed separately for two cases, namely \emph{decoupled} and \emph{coupled} HV antenna clusters, in the following sub-sections.

\vspace{-20pt}
\subsection{Case 1: Decoupled HV Antenna Clusters} \label{sec:KarraysDifferent}
\vspace{-5pt}

Consider the case of decoupled HV antenna clusters via transmission of orthogonal waveform sets over different clusters. As a result, the SV is capable of  grouping  detected paths according to their originating clusters. This  simplifies the HV-sensing in the sequel by building on the techniques in the preceding section.

Recall that  four HV antenna clusters are  located at the vertices of a rectangle with length $L$ and width $W$ that represents the HV shape (see Fig. \ref{fourAnteLS}). The vertex locations are represented as $\{\bp^{(k)} = (x^{(k)}, y^{(k)})^{\mathrm{T}}\}_{k=1}^K$.
Different orthogonal waveform set is assigned to each cluster, allowing SV with prior knowledge on the waveform sets to differentiate the signals transmitted by different clusters. The more challenging case where all clusters are assigned an identical waveform set is studied in the next section.  Let each path be ordered based on HV array index such that $\mathcal{P} = \{\mathcal{P} ^{(1)}, \mathcal{P} ^{(2)}, \mathcal{P} ^{(3)}, \mathcal{P} ^{(4)}\}$ where
$\mathcal{P} ^{(k)}$ represents the set of received signals from the $k$-th HV antenna cluster.
Note that the vertices determines  the HV size and their centroid that gives the HV position. Therefore, the \textbf{sensing problem} can be  represented as
\vspace{-10pt}
\begin{equation}\label{Eq:ProbForm4}
\begin{aligned}
\bigcup\nolimits_{k =1 }^4 \bigcup\nolimits_{p \in \mathcal{P}^{(k)} } \{\theta_p, \phi_p, \lambda_p\} \Rightarrow \{\{\bp^{(k)}\}_{k=1}^4,  \omega\}.
\end{aligned}
\end{equation}
A naive  sensing approach  is to exploit the orthogonality of multiple waveform sets to  decompose the sensing problem into separate positioning of HV antenna clusters using the technique designed in the preceding section.  In the following subsection, we propose a more efficient sensing technique exploiting the prior knowledge of the HV clusters' rectangular topology.

\begin{figure}[t]
\centering
\includegraphics[width=7.5cm]{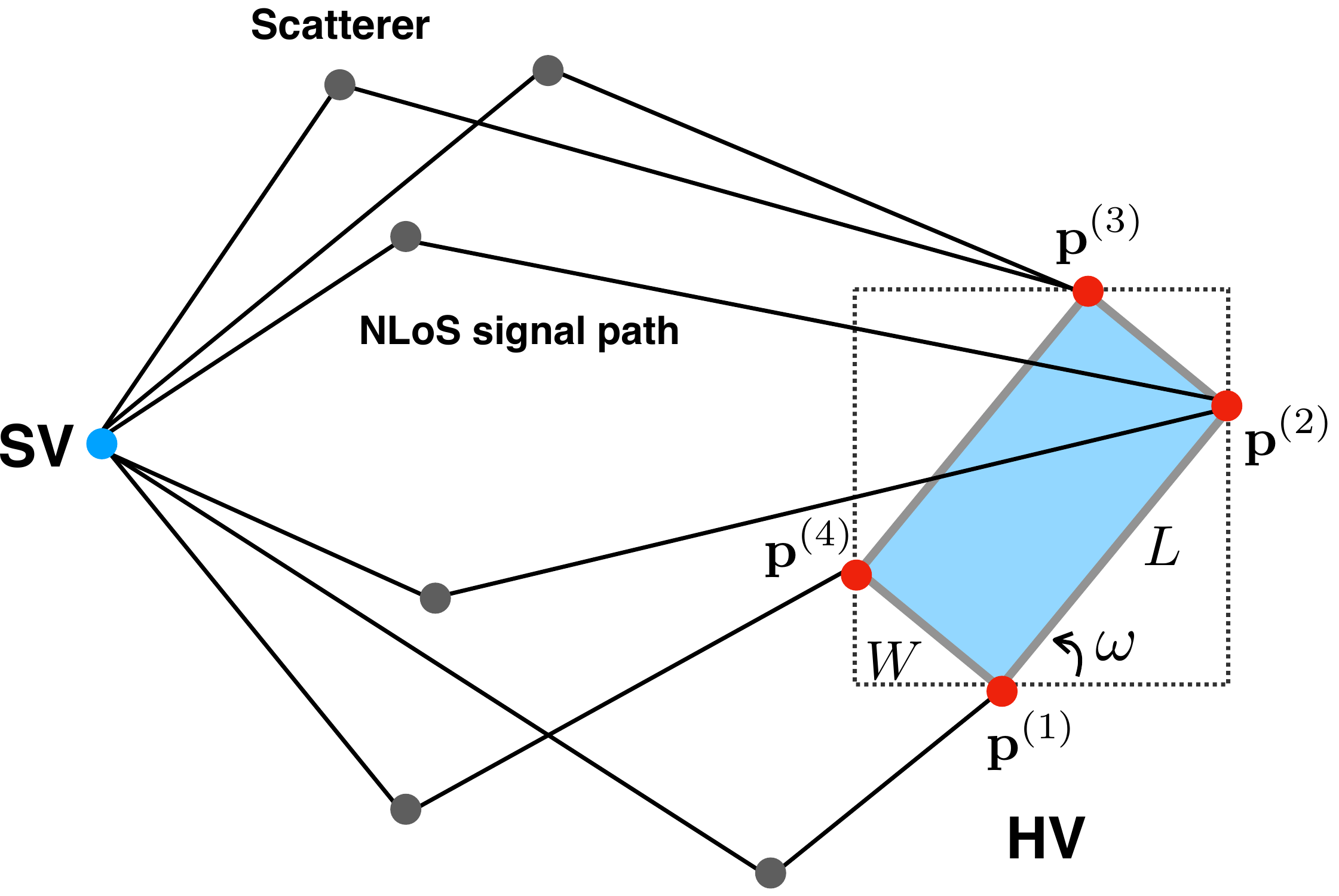}
\vspace{-10pt}
\caption{Rectangular configuration   of  a $4$-cluster HV array and the corresponding multi-path propagations.}\label{fourAnteLS}
\vspace{-30pt}
\end{figure}

\subsubsection{Sensing Feasibility Condition}
Without loss of generality,   assume that the received signal from the first HV antenna cluster, indexed by the set
$\mathcal{P}^{(1)}$,  is not empty and $1\in \mathcal{P}^{(1)}$.
Based on  the rectangular configuration  of $\{\mathbf{p}^{(k)}\}_{k=1}^4$ (see Fig.~\ref{fourAnteLS}),
a  system of equations is formed:
\vspace{-10pt}
\begin{align}\label{FGrelation_4MIMO}\tag{E4}\!\!\!\!\!\!\!\!
\begin{cases}
\nu_{p} \cos(\theta_{p}) -\! (d_{p} - \nu_{p}) \cos(\varphi_{p} + \omega)+\eta_p(\omega,L, W) = \nu_{1} \cos(\theta_{1}) - (d_{1} - \nu_{1}) \cos(\varphi_{1} + \omega),  \\
\nu_{p} \sin(\theta_{p}) - \!(d_{p} - \nu_{p}) \sin(\varphi_{p} + \omega)+\zeta_p(\omega, L, W )  =\nu_{1} \sin(\theta_{1}) - (d_{1} - \nu_{1}) \sin(\varphi_{1} + \omega),
\end{cases}
\end{align}
where $\eta_p(\omega, L, W) = 0$, $L\cdot \cos(\omega)$, $L\cdot \cos(\omega)-W\cdot \sin(\omega)$, and $-W\cdot\sin(\omega)$ when $p\in\mathcal{P}^{(1)}$, $p\in\mathcal{P}^{(2)}$, $p\in\mathcal{P}^{(3)}$, and $p\in\mathcal{P}^{(4)}$, respectively.
$\zeta_p(\omega, L, W)$ is obtained via replacing all $\cos$ and $\sin$ in $\eta_p(\omega, L, W)$ with $\sin$ and $-\cos$, respectively. The number of signal paths is given as $P=|\mathcal{P}|=\sum_{k=1}^4|\mathcal{P}^{(k)}|$.  Compared with \eqref{FGrelation}, the number of equations in \eqref{FGrelation_4MIMO} is the same as $2(P-1)$ while the number of unknowns increases from $(P+2)$ to $(P+4)$ since $L$ and $W$ are also unknown. Consequently, \eqref{FGrelation_4MIMO} has a unique solution when $2(P-1)\geq P+4$.

\vspace{-10pt}
\begin{proposition}[Sensing Feasibility Condition]\label{pro:minNumPath_4MIMO}
\emph{Consider the scenario that the HV is provisioned with a $4$-cluster array and  orthogonal waveform sets are transmitted from different clusters. To sense the position, size, and orientation of the  HV, at least six paths are required: $P\geq 6$.}
\end{proposition}
\vspace{-5pt}
\vspace{-20pt}
\begin{remark}[Advantage of Array-Topology  Knowledge]\emph{The separate positioning  of individual HV antenna clusters  requires at least $16$ NLoS paths (see Proposition \ref{pro:minNumPath}). On the other hand, the proposed sensing technique reduces the number of required paths to only $6$ by exploiting  the prior knowledge of the rectangular configuration  of antenna clusters.}
\end{remark}

\vspace{-10pt}
\subsubsection{Hidden Vehicle Sensing}
Consider the case where channel noise is negligible. The system of equations in \eqref{FGrelation_4MIMO} can be  rewritten in  a compact  matrix form:
\vspace{-10pt}
\begin{align}\label{ax=b_4MIMO}\tag{E5}
\hat{\mathbf{A}}(\omega) \hat{\mathbf{z}}= {\mathbf{B}}(\omega),
\end{align}
where $\hat{\mathbf{z}} =(\mathbf{v}, d_1, L, W)^{\mathrm{T}} \in \mathds{R}^{(P+3)\times 1}$ with $\mathbf{v}$
 following the index ordering of $\mathcal{P}$, and ${\mathbf{B}}(\omega)$ is given in \eqref{matrixB}. The  matrix $\hat{\mathbf{A}}(\omega)$ can be decomposed as
 \vspace{-10pt}
\begin{align}\label{matrixA_4MIMO}
\hat{\mathbf{A}}(\omega) &=
\begin{bmatrix}
  \mathbf{A}(\omega)&\mathbf{L}(\omega)&\mathbf{W}(\omega)
\end{bmatrix}
\in \mathds{R}^{2(P-1) \times (P+3)},
\end{align}
where  $\mathbf{A}(\omega)$ follows  \eqref{matrixA}. Moreover,  $\mathbf{L}(\omega)\in \mathds{R}^{2(P-1) \times 1}$ is given as  $[\mathbf{L}^{(\cos)}(\omega), \mathbf{L}^{(\sin)}(\omega)]^{\mathrm{T}}$ with
\vspace{-10pt}
\begin{align}
\mathbf{L}^{(\cos)}(\omega) &= [\underbrace{0,\cdots,0}_{|\mathcal{P}^{(1)}|-1}, \underbrace{-\cos(\omega),\cdots, -\cos(\omega)}_{|\mathcal{P}^{(2)}|+|\mathcal{P}^{(3)}|}, \underbrace{0,\cdots,0}_{|\mathcal{P}^{(4)}|}  ]^{\mathrm{T}},
\end{align}
where $|\mathcal{P}^{(k)}|$ counts the number of elements in $\mathcal{P}^{(k)}$ and $\mathbf{L}^{(\sin)}(\omega)$ is obtained by replacing all $\cos(\omega)$ in $\mathbf{L}^{(\cos)}(\omega)$ with $\sin(\omega)$. $\mathbf{W}(\omega)$ can be written as  $[\mathbf{W}^{(\sin)}(\omega), \mathbf{W}^{(\cos)}(\omega)]^{\mathrm{T}}$ where
\vspace{-10pt}
\begin{align}
\mathbf{W}^{(\sin)}(\omega) &= [\underbrace{0,\cdots,0}_{|\mathcal{P}^{(1)}|+|\mathcal{P}^{(2)}|-1}, \underbrace{\sin(\omega),\cdots, \sin(\omega)}_{|\mathcal{P}^{(3)}|+|\mathcal{P}^{(4)}|}]^{\mathrm{T}},
\end{align}
and $\mathbf{W}^{(\cos)}(\omega)$ is obtained by replacing all $\sin$ in $\mathbf{W}^{(\sin)}(\omega)$ with $-\cos$.

\noindent \underline{1) Computing $\omega^*$:}
Noting that (\ref{ax=b_4MIMO}) is over-determined when $P\geq 6$, the resultant discriminant of the orientation $\omega$ is similar to that in Proposition~\ref{pro:discriminant} and given as follows.

\vspace{-10pt}
\begin{proposition}[Discriminant of Orientation]\label{pro:discriminantMultiCluster}
\emph{With $P \geq 6$, the unique $\omega^*$ exists when $\hat{\mathbf{B}}(\omega^*)$ is orthogonal to the null column space of $\hat{\mathbf{A}}(\omega^*)$ denoted by $\mathsf{null}(\hat{\mathbf{A}}(\omega^*)^{\mathrm{T}})\in \mathds{R}^{2(P-1)\times (P+1)}$:
\vspace{-10pt}
\begin{align}\label{Eq:OrthogonalNullSpaceMultiCluster}
\mathsf{null}(\hat{\mathbf{A}}(\omega^*)^{\mathrm{T}})^{\mathrm{T}}\hat{\mathbf{B}}(\omega^*)=\mathbf{0}.
\end{align}}
\end{proposition}
\noindent 
Given  this discriminant, a simple 1D search can be performed over the range $[0, 2\pi]$ to find  $\omega^*$.

\noindent \underline{2) Computing $\{\bp^{(k)}\}_{k=1}^4$:}
Given the $\omega^*$, (\ref{ax=b_4MIMO}) can be  solved by
\vspace{-10pt}
\begin{align}\label{LSresult4MIMO}
\hat{\mathbf{z}}^* = \l[ \hat{\mathbf{A}}(\omega^*)^{\mathrm{T}} \hat{\mathbf{A}}(\omega^*)  \r]^{-1} \hat{\mathbf{A}}(\omega^*)^{\mathrm{T}} \hat{\mathbf{B}}(\omega^*).
\end{align}
The positions of  HV antenna clusters,  say $\{\bp^{(k)}\}_{k=1}^4$,  can be  computed  by substituting \eqref{Eq:OrthogonalNullSpaceMultiCluster} and (\eqref{LSresult4MIMO}) into \eqref{Eq:TDoAandPropagationDist} and \eqref{FGrelation_4MIMO}. Extending the above sensing  technique  to the case with channel  noise is  straightforward by modifying \eqref{Eq:OrthogonalNullSpaceMultiCluster} to a minimization problem as in Section~\ref{subsection:NoiseHVSensing}.

\vspace{-20pt}
\subsection{Case 2: Coupled   HV Antenna Clusters}\label{sec:KarraysSame}
\vspace{-5pt}
It is desired to reduce the number of orthogonal waveform sets used by a HV so as to facilitate multi-access by dense HVs. Thus, in this section, we consider the resource-limited  case of coupled HV antenna clusters where  a identical waveform set is shared and transmitted by all HV antenna clusters. The design of HV-sensing is more challenging since the SV is incapable of grouping the signal paths according to their originating HV antenna clusters. For tractability, the objectives of HV-sensing for this scenario is redefined as: 1) positioning of the centroid of HV multi-cluster array  denoted by $\bp_{0}=(x_0, y_0)$; 2) sensing the HV size by estimating the   maximum distance  between HV antenna clusters and  $\bp_{0}$, denoted by $R=\max_{k}{|\bp^{(k)}-\bp_0|}$; 3) estimating the HV orientation $\omega$. It follows that the \textbf{sensing problem} can be formulated as
\vspace{-10pt}
\begin{equation}
\begin{aligned}
 \quad  \bigcup\nolimits_{p \in \mathcal{P}} \{\theta_p, \phi_p, \lambda_p\} \Rightarrow \{\bp_{\textrm{0}}, R, \omega\}.
\end{aligned}
\end{equation}
To solve the problem, we adopt the following two-step approach:

\noindent \textbf{Step 1}:  By assuming that  all signals received at SV originate from the same transmitting location, it is treated  as the HV array  centroid and estimated together with the orientation $\omega$  using the technique in Section \ref{sec:singleCluster}.

\noindent \textbf{Step 2}: Given  $\omega$ and $\bp_{0}$, the size parameter $R$ can be estimated by solving  optimization problems based on bounding the HV array by either a disk or a box.

The techniques for Step 2 are designed in following sub-sections.

\begin{figure*}[t]
\centering
\subfigure[Sensing disk (2D).]{\includegraphics[width=4cm]{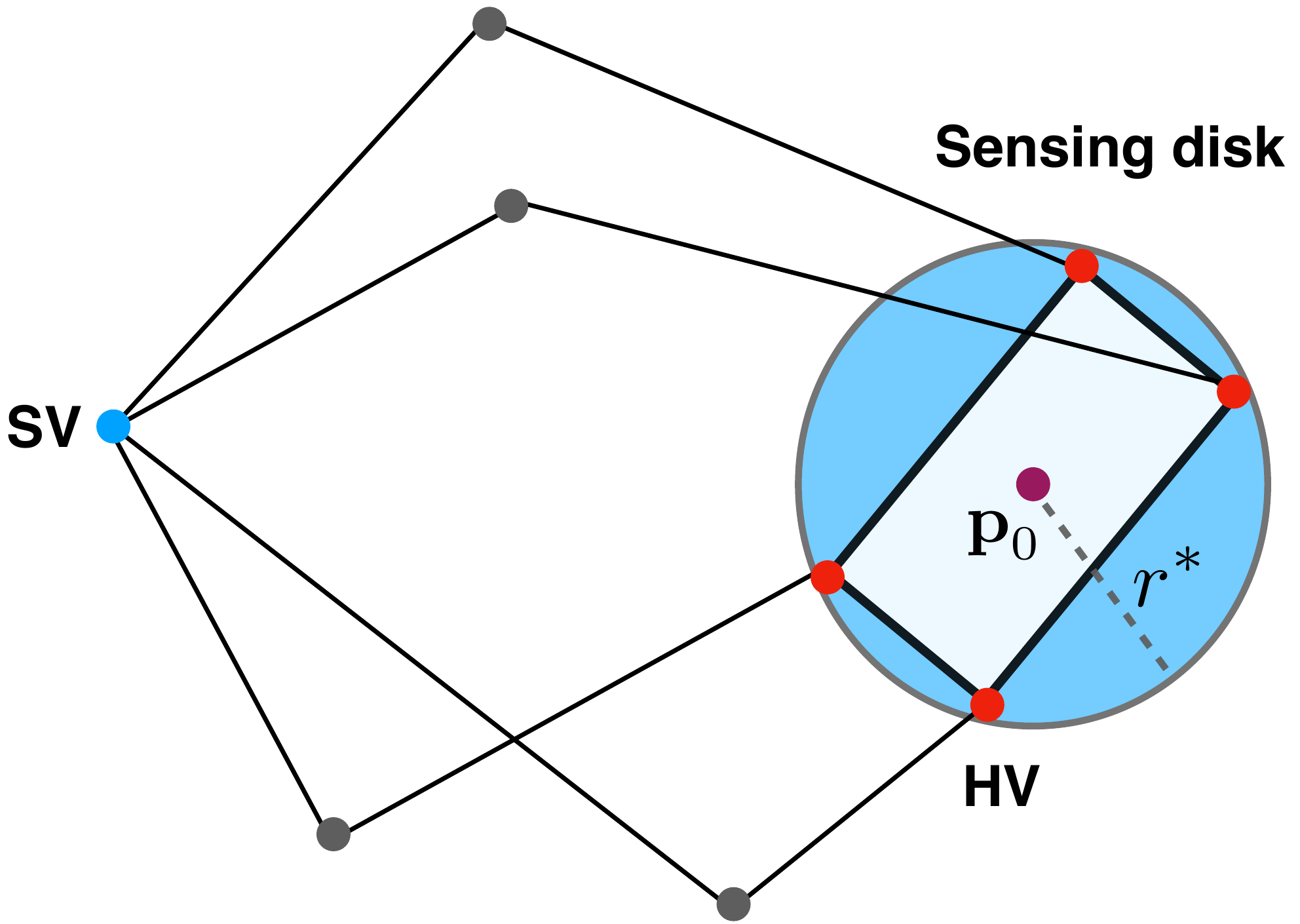}}
\subfigure[Sensing sphere (3D).]{\includegraphics[width=4cm]{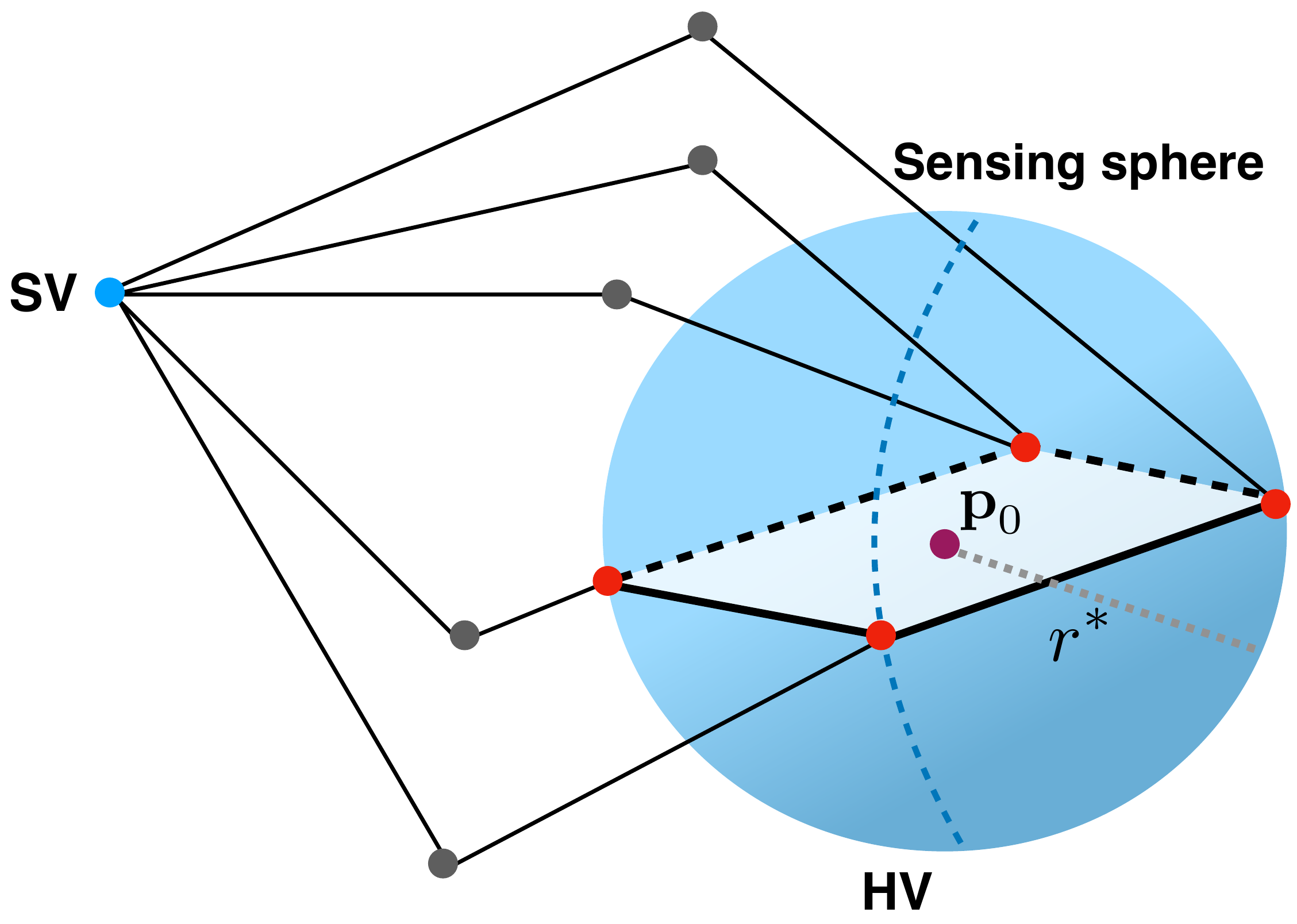}}
\subfigure[Sensing box  (2D).]{\includegraphics[width=4cm]{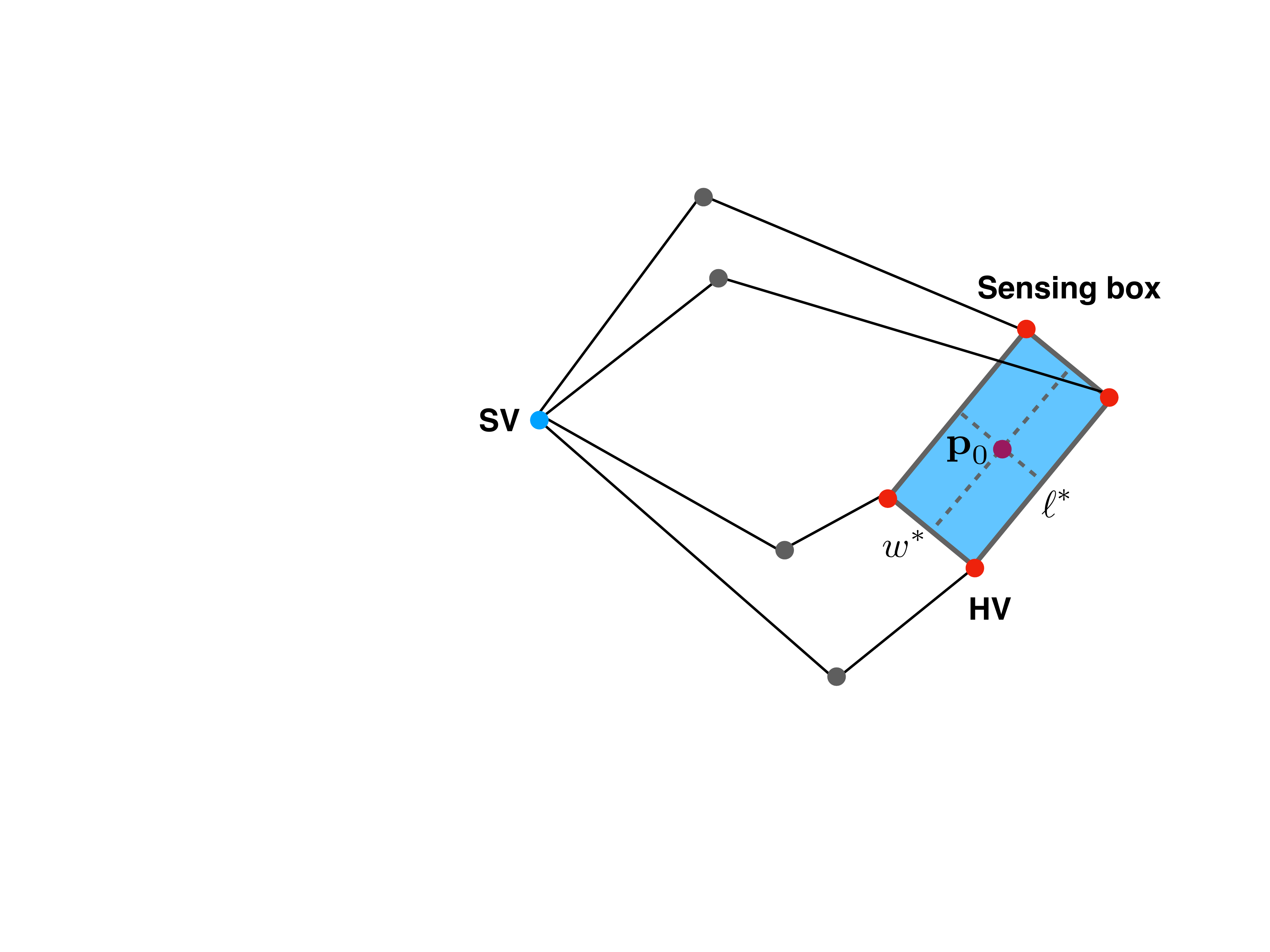}}
\subfigure[Sensing cuboid (3D).]{\includegraphics[width=4.1cm]{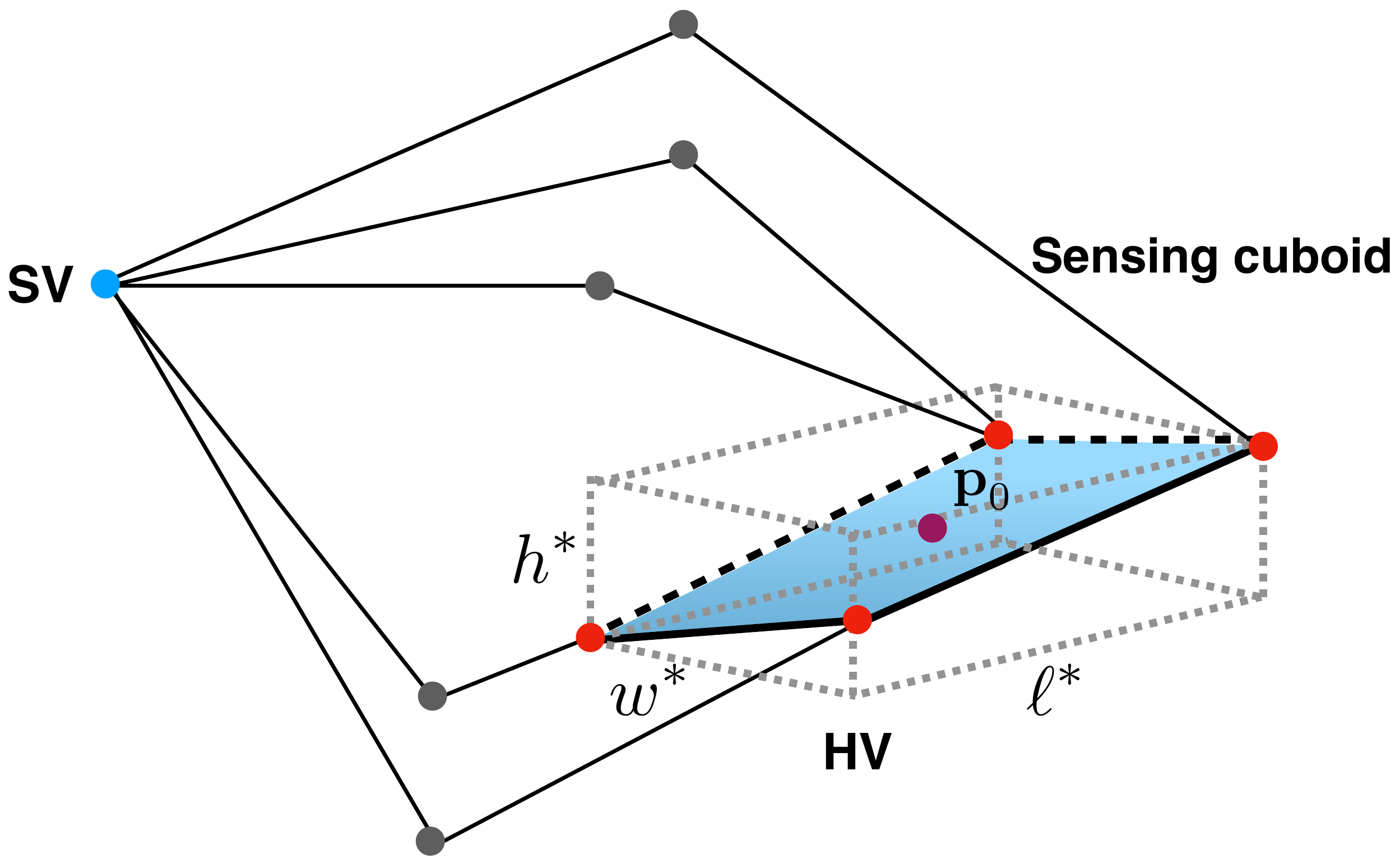}}
\caption{Different approaches of HV-sensing for the scenario of a 4-cluster HV array and an identical   waveform set for transmission by different HV antenna clusters.}\label{noIndexEst}
\vspace{-30pt}
\end{figure*}

\subsubsection{HV Size Sensing by Disk Minimization}\label{sec:SCM}
Note that the HV array is outer bounded by a disk. Then the problem of estimating the HV size parameter $R$ at SV can be translated into the  optimization problem of minimizing the bounding-disk radius.
As shown in Fig. \ref{noIndexEst}(a),
we define a \emph{sensing disk} $\mathcal{C}(\bp_{0}, r)$  centered at the estimated centroid $\bp_{0}$ with the  radius $r$:
\vspace{-10pt}
\begin{align}\label{sensingDisk}
\mathcal{C}(\bp_{0}, r)=\left\{(x, y)|\l(x - x_{0} \r)^2 + \l(y - y_{0} \r)^2 \leq r^2\right\}. 
\end{align}
A constraint is applied that all HV antennas, or equivalently the origins of all signal paths received at the SV,  should lie within the disk. Then estimating the HV size $R$ can be translated into the following problem of \emph{disk  minimization}:
\vspace{-10pt}
\begin{equation}\label{CircleMinProb}\tag{E6}
\begin{aligned}
R = &\min_{ d_1, r, \{\nu_p,x_p,y_p\} } r \\
\text{s.t.}&\quad  \l(x_p - x_{0} \r)^2 + \l(y_p- y_{0} \r)^2\leq r^2, \quad  0 < \nu_p < d_1+c\rho_p,  \\
&\quad (x_p, y_p) \text{ satisfies \eqref{TXposition2D} with } d_p=d_1+c\rho_p, ~~\forall p \in \mathcal{P},
\end{aligned}
\end{equation}
where the first constraint is as mentioned above and  the second represents the distance after the reflection $\nu_p$ cannot exceed the total propagation distance $d_p$ represented in terms of  $d_1$ and TDoA $\rho_p$ as $d_p=d_1+c\rho_p$ with $\{\rho_p\}$ being the TDoAs [see \eqref{Eq:TDoAandPropagationDist}]. The values of $\{x_p,y_p\}$ are directly calculated by plugging the optimized $d_1$ and $\{\nu_p\}$ into \eqref{FGrelation}, corresponding to the optimal radius $r^*$ according to the first constraint of \eqref{CircleMinProb}. One can observe that Problem~\eqref{CircleMinProb} is a problem of \emph{second-order cone programming} (SOCP). Thus, it is a convex optimization problem and can be efficiently solved numerically e.g., using the efficient MatLab toolbox such as CVX.

Analyzing the problem structure can shed light on the number of required paths for HV-sensing in the current scenario. The existence and uniqueness of the optimal solution $r^*$ for Problem~\eqref{CircleMinProb}  can be explained intuitively by considering the feasible range of  $d_1$. Let $\mathcal{S}_{p}(r)$ represent the feasible range of the optimization variable $d_1$ for path $p$ when the disk radius is given as $r$:
\vspace{-10pt}
\begin{align}\label{setofD1_SCM_Individual}
\mathcal{S}_{p}(r) &= \l\{ d_1\l| \text{all constraints for  path $p$ in \eqref{CircleMinProb}} \r.\r\}.
\end{align}
 Then the feasible range of $d_1$, denoted by $\mathcal{S}(r)$, is the intersection of the feasible range of $d_1$ for every path $p$, i.e., $\mathcal{S}(r) =\bigcap_{p \in \mathcal{P}} \mathcal{S}_{p}(r)$. This is because all the paths share the same $d_1$ and thus the feasible range of $d_1$ should satisfy all paths' constraints in \eqref{CircleMinProb} simultaneously.
Next, it is straightforward to show the following monotonicity of $\mathcal{S}(r)$:  $\mathcal{S}(r_1)\subseteq \mathcal{S}(r_2)$ if $r_1\leq r_2$ with  $\mathcal{S}(0)=\emptyset$. Based on the monotonicity,  there always exists an  optimal and unique solution  $r^*$ for Problem~\eqref{CircleMinProb} such that $\mathcal{S}(r) \neq  \emptyset$  if $r\geq r^*$ or otherwise $\mathcal{S}(r) =  \emptyset$. In other words,
\vspace{-10pt}
\begin{align}\label{Optimal_r}
r^* = \inf \l\{ r>0\l| \mathcal{S}(r) \neq \emptyset\r. \r\}=\sup\l\{ r >0\l| \mathcal{S}(r) = \emptyset \r.\r\}.
\end{align}
The value $r^*$ corresponds to the critical case where there exist two feasible range sets $\mathcal{S}_{p}(r^*)$ and $\mathcal{S}_{p'}(r^*)$ only contact each other at their boundaries such that $\mathcal{S}(r)$ contains a single feasible point $d_1^*$ that corresponds to $r^*$. This leads to the following proposition.

\vspace{-10pt}
\begin{proposition}[HV Size Sensing  by Disk Minimization]\label{pro:optimaRD1_SCM}
\emph{Given the solution $r^*$ for Problem~\eqref{CircleMinProb},  there always exist at least two paths, say $p_1$ and $p_2$, whose originating positions lie  on the boundary of  the minimized disk  $\mathcal{C}(\bp_0, r^*)$:
\vspace{-10pt}
\begin{align}\label{Optimal_Structure:SCM}
\l(x_{p_1} - x_{0} \r)^2 + \l(y_{p_1} - y_{0} \r)^2
=\l(x_{p_2} - x_{0} \r)^2 + \l(y_{p_2} - y_{0} \r)^2= (r^*)^2.
\end{align}
}
\end{proposition}
\vspace{-10pt}

Instead of the earlier intuitive argument, Proposition \ref{pro:optimaRD1_SCM} can be proved rigorously using the \emph{Karush-Kuhn-Tucker} (KKT) conditions as shown in Appendix~\ref{proof:Pro:minPosiCircle}.

\vspace{-10pt}
\begin{remark}[Feasible Condition of HV Sensing by Disk Minimization]\label{Re:Feasible}\emph{Though two paths are required to determine the optimal disk radius $R = r^*$ based on Proposition \ref{pro:optimaRD1_SCM},  at least four paths are required for  estimating  the required  centroid  $\bp_0$ (see Proposition~\ref{pro:minNumPath}). }
\end{remark}
\vspace{-5pt}

\vspace{-15pt}
\begin{remark}[Extension to 3D Propagation]\label{3DPCM}
\emph{The extension to 3D propagation model in Section~\ref{Sec:3DSignalModel} is straightforward by using a sphere instead of a disk [see Fig. \ref{noIndexEst}(b)]. The  resultant  sphere minimization problem  has the same form as  Problem~\eqref{CircleMinProb} except that  the first constraint  modified as $\l({x}_{p} - {x}_{0} \r)^2 + \l({y}_{p} - {y}_{0} \r)^2 + \l({z}_{p} - {z}_{0} \r)^2 \leq r^2, ~~\forall p \in \mathcal{P}$,
where the centroid $\bp_0=({x}_{0}, {y}_{0}, {z}_{0})$ is estimated using the technique in Section~\ref{Sec:3DSignalModel}. Again, the problem can be optimally solved since it still follows SOCP structure. }
\end{remark}

\vspace{-10pt}
\subsubsection{HV Size Sensing by Box Minimization}\label{sec:SBM}
In the preceding sub-section, the HV size is estimated by bounding the HV array by a disk and then minimizing it. In this sub-section, the disk is replaced by a box and the HV size estimation is translated into the problem of box minimization. Compared with disk minimization, the current technique improves the estimation accuracy since a vehicle typically has a rectangular shape. Let $L$ and $W$ be the length and width of the rectangular where the HV antenna clusters are placed at its vertices (see Fig. \ref{noIndexEst}(c)). Then the problem of HV size sensing is to estimate both $L$ and $W$. Recall that the HV array centroid $\bp_0$ and orientation $\omega$ are estimated in \textbf{Step 1} of the proposed sensing approach as mentioned.  Given $\bp_0$ and $\omega$, we define a \emph{sensing box} for bounding the HV array, denoted as  $\mathcal{B}(\bp_{0}, \omega, \ell, w)$,  as an $\omega$-rotated rectangle centered at  $\bp_{0}=(x_0, y_0)$ and having the  length $\ell$ and width $w$:
\vspace{-10pt}
\begin{align}\label{Eq:SizingBox}
&\mathcal{B}(\bp_0, \omega, \ell, w) =\left\{(x, y)\left|-\frac{1}{2}
\begin{bmatrix}
\ell, w
\end{bmatrix}^{\mathrm{T}}
 \preceq
\bR(\omega) \begin{bmatrix}
  x - x_{0}, y - y_{0}
\end{bmatrix}^{\mathrm{T}} \preceq \frac{1}{2}
\begin{bmatrix}
\ell, w
\end{bmatrix}^{\mathrm{T}} \right.\right\},
\end{align}
 where $\bR(\omega)$ is the counterclockwise rotation matrix with the rotation angle  $\omega$ given as
 \vspace{-10pt}
 \begin{align}
\bR(\omega)=  \begin{bmatrix}
  \cos(\omega)&\sin(\omega)\\
  -\sin(\omega)&\cos(\omega)
\end{bmatrix},
 \end{align}
and $\preceq$ represents an element-wise inequality.
Like disk minimization  in the previous subsection, finding the correct $L$ and $W$ is transformed into the following \emph{box minimization} problem:
\vspace{-10pt}
\begin{equation}\label{BoxMinProb}\tag{E7}
\begin{aligned}
\{L, W\} &=\arg\min_{ d_1, \ell, \omega, \{\nu_p, x_p, y_p\} } (\ell^2+w^2) \\
\text{s.t.}\quad &
-\frac{1}{2}
\begin{bmatrix}
\ell, w
\end{bmatrix}^{\mathrm{T}}
 \preceq
\bR(\omega) \begin{bmatrix}
  x_p - x_{0}, y_p - y_{0}
\end{bmatrix}^{\mathrm{T}} \preceq \frac{1}{2}
\begin{bmatrix}
\ell, w
\end{bmatrix}^{\mathrm{T}}, ~~\forall p \in \mathcal{P},\\
&\quad 0 < \nu_p < d_1+c\rho_p, ~~\forall p \in \mathcal{P},\\
\end{aligned}
\end{equation}
where the first constraint represents that all origins of signal paths $\{x_p, y_p\}$ should be inside $\mathcal{B}(\bp_0, \omega, \ell, w)$ defined in \eqref{Eq:SizingBox} and
the second one is the same as in \eqref{CircleMinProb}. Problem \eqref{BoxMinProb} can be solved by  \emph{quadratic programming} (QP), which is a convex optimization problem and can be efficiently solved using a  software  toolbox such as MatLab CVX. A result similar to that in Proposition \ref{pro:optimaRD1_SCM} can be obtained for HV size sensing by box minimization as shown below.

\vspace{-10pt}
\begin{proposition}[HV Size Sensing  by Box  Minimization]\label{pro:optimaRD1_SBM}
\emph{Given the solution $\{\ell^*, w^*\}$ for Problem \eqref{BoxMinProb},  there always exist at least two paths, say $p_1$ and $p_2$, whose originating positions lie  on two different vertices  of  the minimized box: 
\vspace{-5pt}
\begin{align}
\bR(\omega) \begin{bmatrix}\label{Optimal_structure_SBM}
  x_{p_1}^* - x_{p_2}^*, y_{p_1}^* - y_{p_2}^*
\end{bmatrix}^{\mathrm{T}} =
\begin{bmatrix}
\ell^*, w^*
\end{bmatrix}^{\mathrm{T}}
\mathrm{or}
\begin{bmatrix}
-\ell^*, w^*
\end{bmatrix}^{\mathrm{T}}
\mathrm{or}
\begin{bmatrix}
\ell^*, 0
\end{bmatrix}^{\mathrm{T}}
\mathrm{or}
\begin{bmatrix}
0, w^*
\end{bmatrix}^{\mathrm{T}}.
\end{align}}
\end{proposition}
\begin{proof}
See Appendix~\ref{proof:Pro:minPosiBox}.
\end{proof}

\vspace{-15pt}
\begin{remark}[Feasible Condition of HV-Sensing by Box Minimization]\label{Re:Feasible:1}\emph{A similar remark as Remark~\ref{Re:Feasible} for disk minimization also applies to the current technique. Specifically, though two paths are required to determine the optimal box length $L = \ell^*$ and width $W = w^*$  based on Proposition \ref{pro:optimaRD1_SBM},  at least four paths are required for  estimating  the required HV  centroid  $\bp_0$ and orientation $\omega$  (see Proposition~\ref{pro:minNumPath}). }
\end{remark}
\vspace{-5pt}

\vspace{-15pt}
\begin{remark}[Sensing Box Minimization for Decoupled Antenna Clusters]\label{remark:SBMwithIndex}
\emph{The technique of HV size  sensing by box minimization developed for the case of coupled HV antenna clusters can be also modified for use in the case of decoupled clusters. Roughly speaking, the modified technique involves separation minimization of four boxes corresponding to the positioning  of four antenna clusters. As the modification is straightforward, the details are omitted for brevity. The resultant advantage with respect to the original sensing technique proposed in Section \ref{sec:KarraysDifferent} is to reduce the minimum number of required paths from $6$ (see Proposition \ref{pro:minNumPath_4MIMO}) to $4$. }
\end{remark}

\vspace{-20pt}
\begin{remark}[Extension to 3D Propagation]\label{3DPBM}
\emph{Similar to  Remark \ref{3DPCM} for disk minimization, the technique of HV size sensing by box minimization originally designed for 2D propagation can be extended to  3D propagation model by using a cuboid instead of a  box, yielding the problem of  \emph{cuboid minimization} as illustrated in  Fig. \ref{noIndexEst}(d). Compared with \eqref{BoxMinProb},  the objective function of the cuboid minimization  is $\ell^2+w^2+h^2$ where the new variable $h$ is added to represent the height of the cuboid. In addition, the first constraint in \eqref{BoxMinProb} is modified  as
\vspace{-10pt}
\begin{align}
-\frac{1}{2}
\begin{bmatrix}
\ell, w, h
\end{bmatrix}^{\mathrm{T}}
 \preceq
\bR_{\textrm{3D}}(\omega, \varrho)
\begin{bmatrix}
  {x}_{p} - {x}_{0}, {y}_{p} - {y}_{0}, {z}_{p} - {z}_{0}
\end{bmatrix}^{\mathrm{T}}
\preceq \frac{1}{2}
\begin{bmatrix}
\ell, w, h
\end{bmatrix}^{\mathrm{T}}
,~~\forall p \in \mathcal{P}
\end{align}
where $\bR_{\textrm{3D}}(\omega, \varrho)$ is the 3D counterclockwise rotation matrix with the rotation angles  $\omega$ and $\varrho$ as
\begin{align}
\bR_{\textrm{3D}}(\omega, \varrho)=\begin{bmatrix}\label{lengthCon}
  \cos(\omega)&-\sin(\omega)\cos(\varrho)&\sin(\omega)\sin(\varrho)\\
  \sin(\omega)&\cos(\omega)\cos(\varrho)&-\cos(\omega)\sin(\varrho)\\
  0&\sin(\varrho)&\cos(\varrho)
\end{bmatrix},
\end{align}
and the centroid $\bp_0=({x}_{0}, {y}_{0}, {z}_{0})$ can be obtained by  the technique in Section~\ref{Sec:3DSignalModel}. The cuboid minimization is still QP and the solution approach is similar to  that for the 2D counterpart.}
\end{remark}

\vspace{-20pt}
\begin{remark}[Size Underestimation]\label{remark:SizeUnderestimation}
\emph{The objective functions of disk minimization in \eqref{CircleMinProb} and the box minimization in \eqref{BoxMinProb} are one-dimension distances [i.e., the radius in \eqref{CircleMinProb} and the diagonal distance in \eqref{BoxMinProb}], while two distances, the length and width of vehicles, are required to represent the area of the corresponding vehicle. In other words, the resultant  size estimation may be smaller than the real one, called \emph{underestimation}, which makes it challenging to guarantee auto-driving's safety. However, the underestimation can happen only when the estimated coordinates of all paths, i.e., $\{x_p, y_p\}$, are closer to the centroid due to the first constraint. 
As a result, the resultant sizes are less likely to be underestimated, which is verified by simulation in Fig. \ref{Sizing}(b) in Sec. \ref{sim:Size}.}
\end{remark}


\vspace{-30pt}
\section{Coping with  Insufficient Multi-Path}\label{sec:practicalIssue}
\vspace{-5pt}

The HV-sensing techniques designed in the preceding sections require at least four propagation paths to be effective. In practice, it is possible to happen that  the number of observed (i.e., detectable) paths may  be insufficient, i.e., $P<4$,  due to either sparse scatterers or the fact that most paths are severely attenuated. To address this practical issue, two solutions are proposed in the following sub-sections,  called \emph{sequential path combining} and \emph{random  directional beamforming}. For simplicity, we focus on the case of single-cluster HV array while  the extension to the case of multi-cluster array is straightforward.

\vspace{-20pt}
\subsection{Sequential Path Combining}\label{sec:SequentialCombine}
\vspace{-10pt}


As shown in Fig. \ref{insuffPaths}, the technique of sequential path combining implemented at the  SV  merges paths from repeated  transmissions of HV till a sufficient  number of paths is identified for the purpose of subsequent  HV-sensing. Let $Q$ denote the number of  HV's repetitive  transmissions with a constant interval denoted by $\Delta$. The interval is chosen to be much larger than the coherence time, enabling SV to differentiate the arrival paths according to their transmission time instants. In addition,  the total transmission duration $Q\Delta$ should be sufficiently small enough to guarantee the constant velocities of HV and SV within the duration. 
Assume  that the relative orientation of  driving direction and velocity  of HV with  respect to  SV, namely  $\omega$ and $v$, remain  constant within  the entire duration of $Q$ intervals $Q\Delta$. Let  $\mathcal{P}_q$ denote the set of observed paths of the $q$-th transmission. Then the following system of equations are formed:
\vspace{-10pt}
\begin{align}
\begin{cases}
\nu_{p} \cos(\theta_{p}) - (d_{p} - \nu_{p}) \cos(\varphi_{p} + \omega)+ v(q-1)   \cos(\omega)= \nu_{1} \cos(\theta_{1}) - (d_{1} - \nu_{1}) \cos(\varphi_{1} + \omega),  \\
\nu_{p} \sin(\theta_{p}) - (d_{p} - \nu_{p}) \sin(\varphi_{p} + \omega)+ v(q-1)  \sin(\omega) =\nu_{1} \sin(\theta_{1}) - (d_{1} - \nu_{1}) \sin(\varphi_{1} + \omega). \nn
\end{cases}
\end{align}
where $ p \in \mathcal{P}_q$ and $q = 1, 2, \cdots, Q$.  They can be solved following a  similar procedure as  in Section~ \ref{sec:KarraysDifferent}. Let $P_{1:q}$ be the total number of paths identified due to  the $q$  transmissions, i.e.,  $P_{1:q}=|\mathcal{P}_1|+|\mathcal{P}_2|+\cdots+ |\mathcal{P}_q|$. Noting that the number of equations above  is  $2(P_{1:q}-1)$ and the number of unknowns are $(P_{1:q}+3)$ including $\{\nu_p\}$, $d_1$, $\omega$ and $v$.
As a result, the condition for the SV collecting sufficient paths for HV sensing is $2(P_{1:q}-1)\leq (P_{1:q}+3)$ or equivalently $P_{1:q} \geq 5$. So path combining over multiple sequential transmissions overcomes the practical  limitation  of insufficient  paths. 

\begin{figure}[t]
\centering
\begin{minipage}{0.45\textwidth}
\centering
\includegraphics[width=7cm]{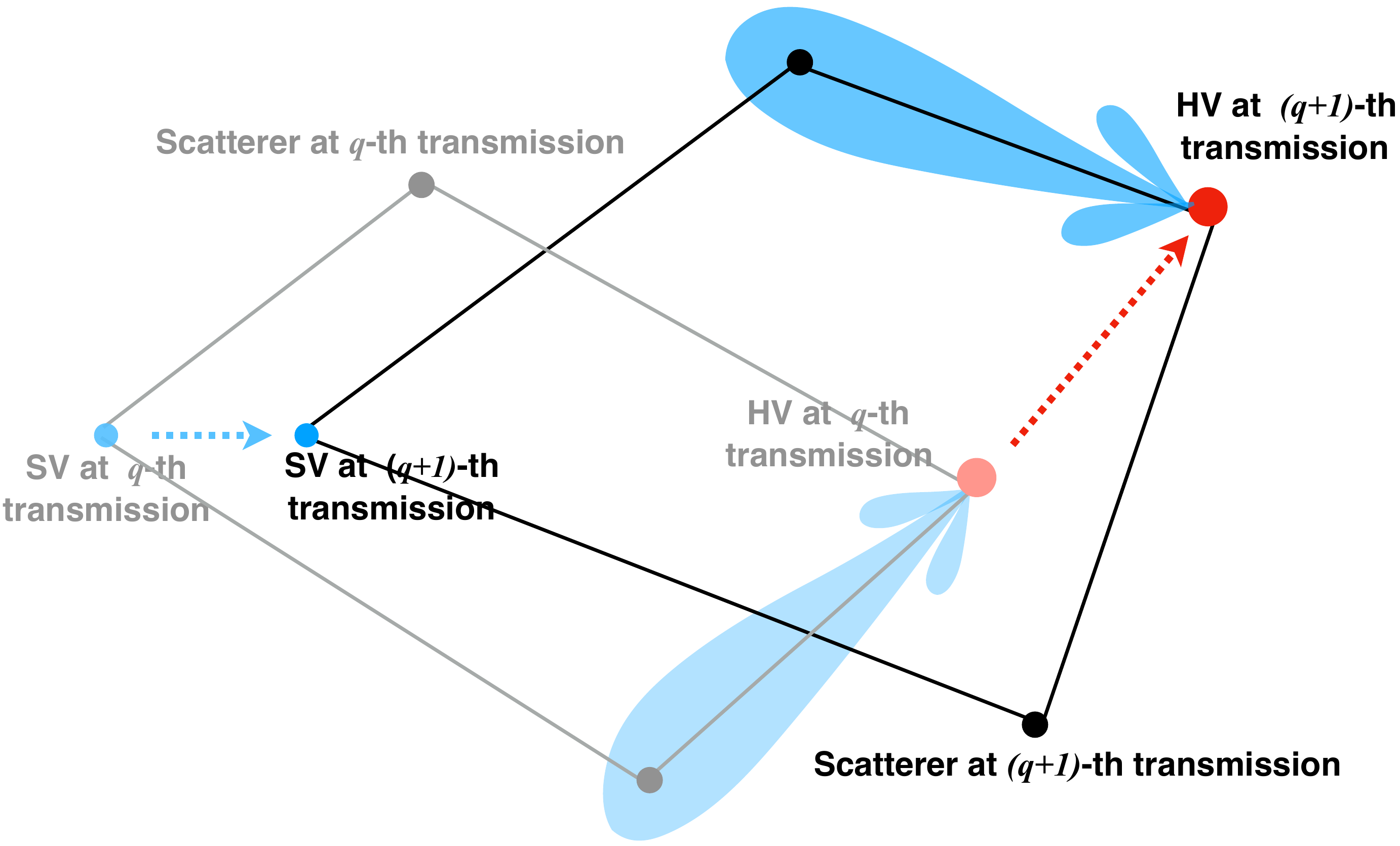}
\caption{Proposed solutions for coping with  insufficient number of propagation  paths.}\label{insuffPaths}
\end{minipage}
\begin{minipage}{0.5\textwidth}
\centering
\subfigure[Highway scenario.]{\includegraphics[width=5cm]{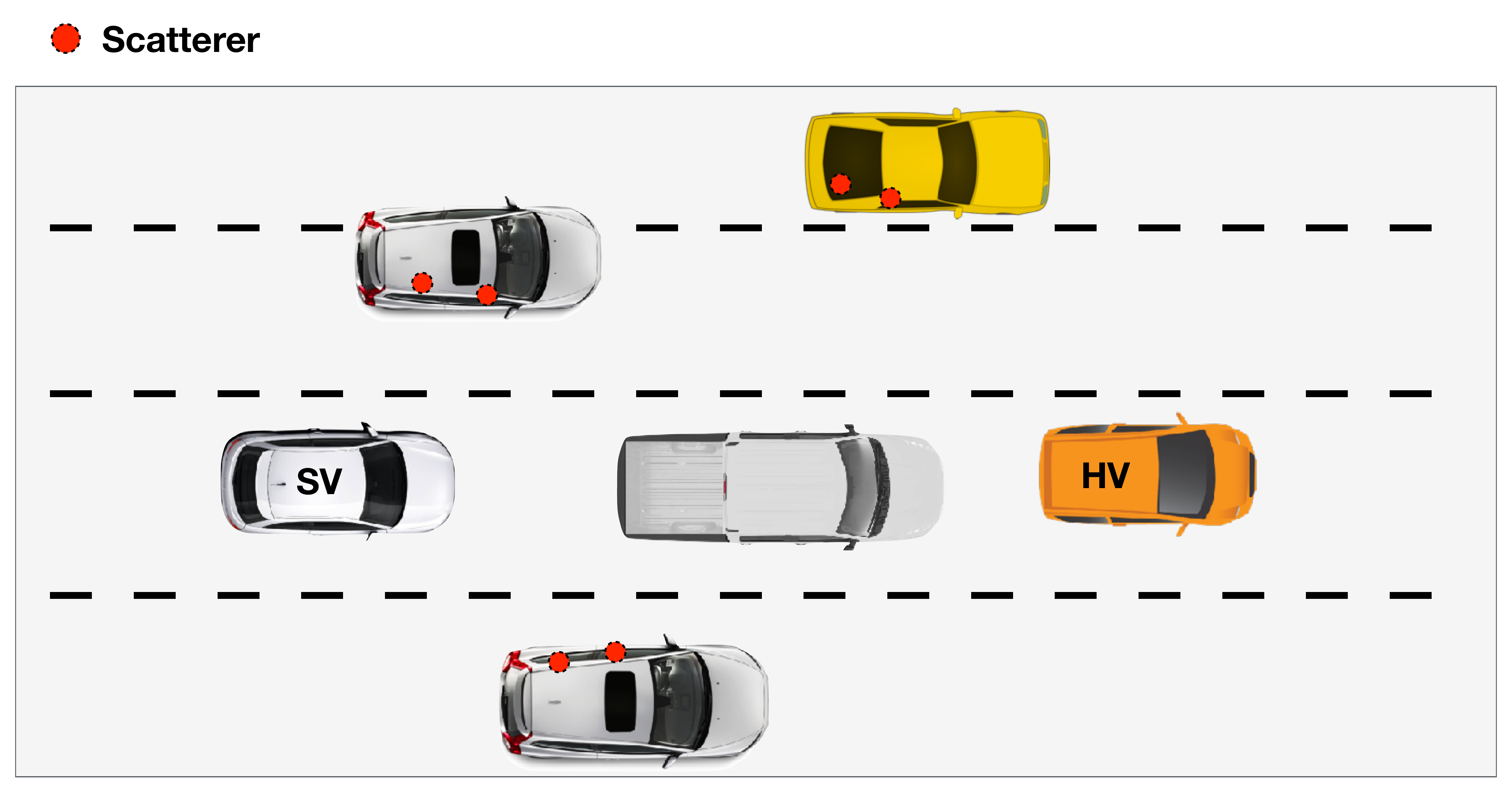}}
\subfigure[Rural scenario.]{\includegraphics[width=5cm]{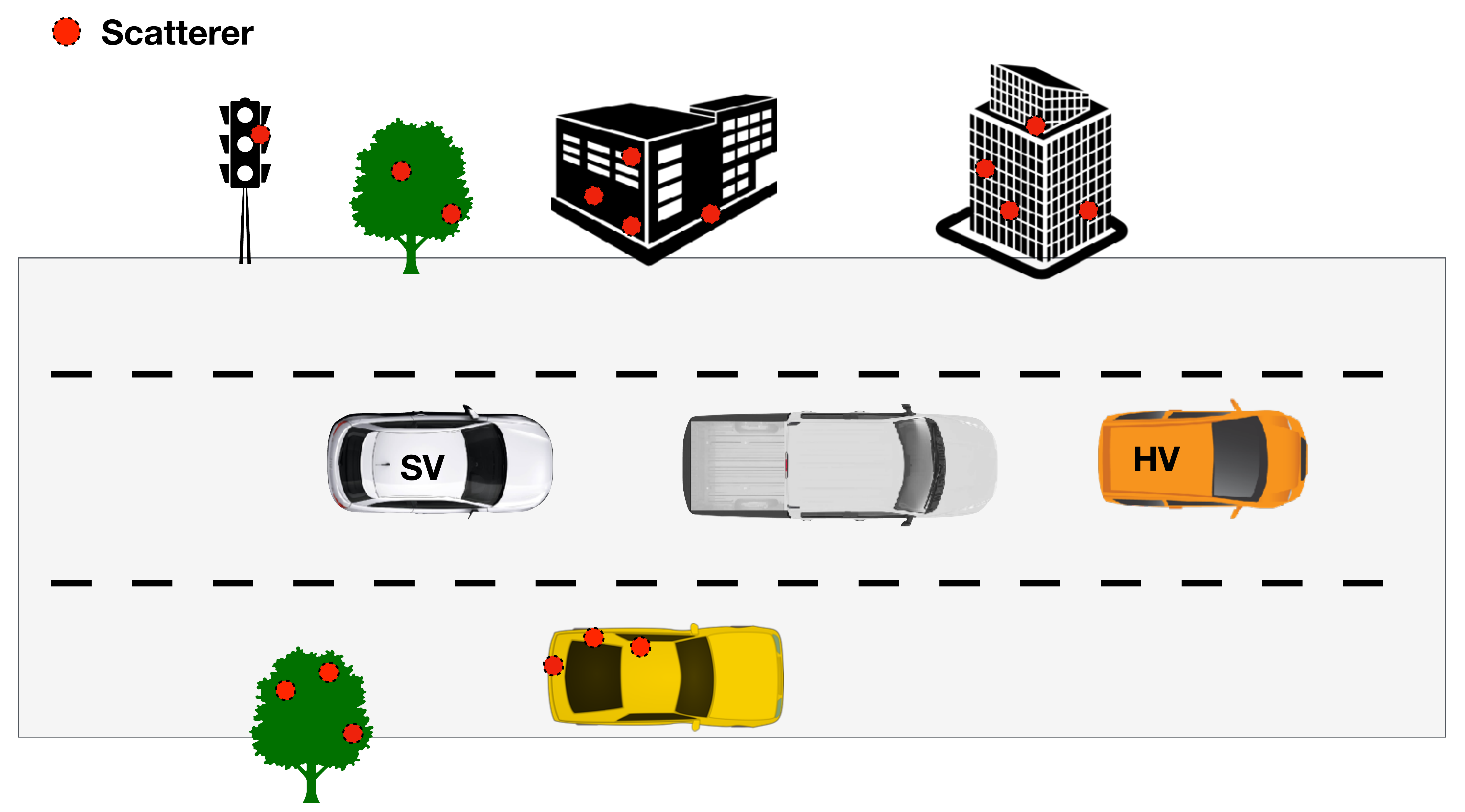}}
\vspace{-10pt}
\caption{Two typical driving scenarios considered in simulation.}\label{highwayRural}
\end{minipage}
\vspace{-35pt}
\end{figure}


\vspace{-20pt}
\subsection{Random Directional Beamforming}  \label{sec:beamforming}
\vspace{-10pt}

To further enhance the effectiveness of sequential path combining, a directional beam can be randomly steered at HV over sequential transmissions. Its purpose is to  reveal some paths that are otherwise hidden to SV due to severe attenuation. The  beam width can be set as ranging from $90^\circ$ to $30^\circ$ with gain ranging from $3$ dB to $10$ dB,  
which helps reach faraway scatterers by focusing the transmission power in their directions and thereby  mitigating  path loss  \cite{viswanath2002opportunistic}. Note that a single trial of  randomly steered beam may not find enough paths. Thus, it is important to combine the technique with sequential path combining designed in the preceding sub-section for the former to be effective. Their integrated operation is illustrated  in Fig.~\ref{insuffPaths} and its effectiveness is  verified by simulation in the sequel.

\vspace{-15pt}
\section{Simulation Results}\label{sec:simulation}
\vspace{-10pt}

In this section, the performance of the proposed vehicular sensing techniques are evaluated by  realistic simulation. Consider the V2V channel model.
We adopt the geometry-based stochastic channel model given in \cite{karedal2009geometry} for modeling the practical scatterers distribution and V2V propagation channel, which has been validated by real measurement data.  Two types of scatterers, namely mobile scatterers (e.g., from moving vehicles) and static scatterers (e.g., from road signs and buildings), are simulated. As illustrated in Fig.~\ref{highwayRural}(a), in the highway scenario, most of the scatterers are mobile scatterers. On the other hand, in the rural scenario illustrated in Fig.~\ref{highwayRural}(b), the scatterers include vehicles on the road as well as relatively denser stationary objects off the road. 
The locations of mobile and static scatterers are randomly distributed over the entire area depending on their densities described in \cite[Table 1]{karedal2009geometry}. Given the scatterer location, the corresponding AoA and AoD are determined without explicit distribution. If more than two scatterer locations are close to each other, the corresponding paths are unresolvable due to similar  AoAs and AoDs. Instead, they are observed as a single signal path with higher power.  HV and SV are assumed to drive in opposite directions. The key simulation parameters and their values are summarized in Table \ref{Table:Simulation} unless stated otherwise. We use MatLab R2015b and its CVX tool box for solving optimization problems and simulations.


\begin{table}[t]
\centering
\caption{Simulation Settings}
\vspace{-10pt}
\setlength{\tabcolsep}{2pt}
\footnotesize
\begin{tabular}{c|c}
\toprule
\textbf{Simulation parameter} &\textbf{Value} \\
\midrule
Carrier frequency $f_c$,  transmission bandwidth $B$ &  $5.9$ GHz, $100$ MHz \\
Number of transmit, receive antennas per cluster $\{M_t, M_r\}$ & $\{20,20\}$ \\
Transmit power, noise power & $23$ dBm,  $-70$ dBm \\
Size of vehicle $L\times W$, inter-vehicle distance & $3\times 6~\textrm{m}^2$, $50~\textrm{m}$ \\
Relative velocity between HV and SV $v$ (driving in opposite directions) & $200~\textrm{km}/\textrm{h}$ \\
\bottomrule
\end{tabular}
\vspace{-10pt}
\label{Table:Simulation}
\end{table}

\vspace{-15pt}
\subsection{Vehicular Positioning}
\vspace{-5pt}
The metric for measuring positioning accuracy is defined as the average Euclidean squared distance of estimated  positions of vehicle antenna clusters   to their true locations: $\frac{1}{4} \sum_{k=1}^{4}\| \bp^{*(k)} - \bp^{(k)}\|^2$, named \emph{average positioning error}. Note that the metric also indirectly measures the accuracy of estimated vehicle size and orientation that are determined by the clusters positions.

\begin{figure}[t]
\centering
\subfigure[Effect of number of observed signal paths.]{\includegraphics[width=8cm]{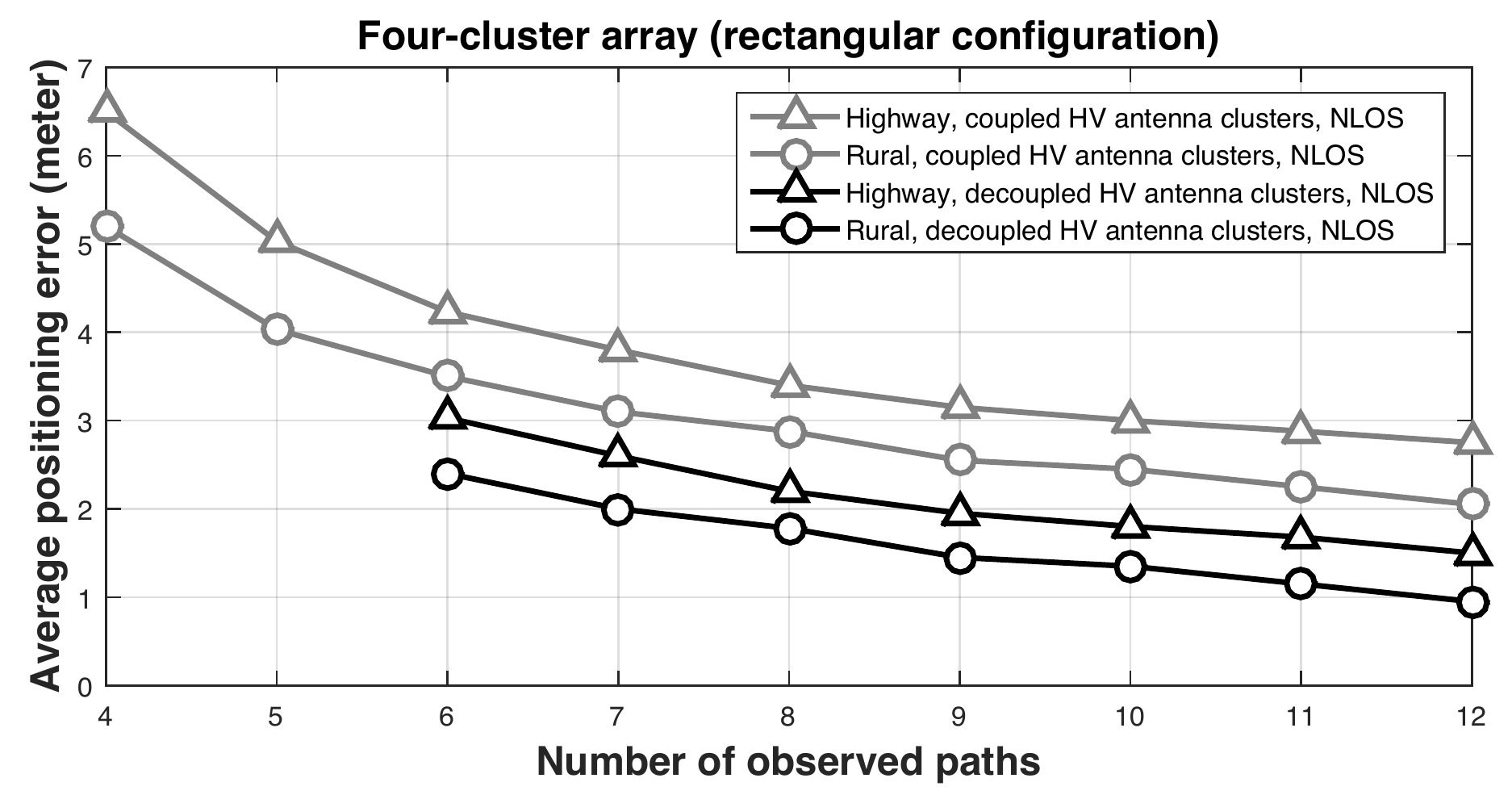}}
\subfigure[Effect of inter-vehicle distance.]{\includegraphics[width=8cm]{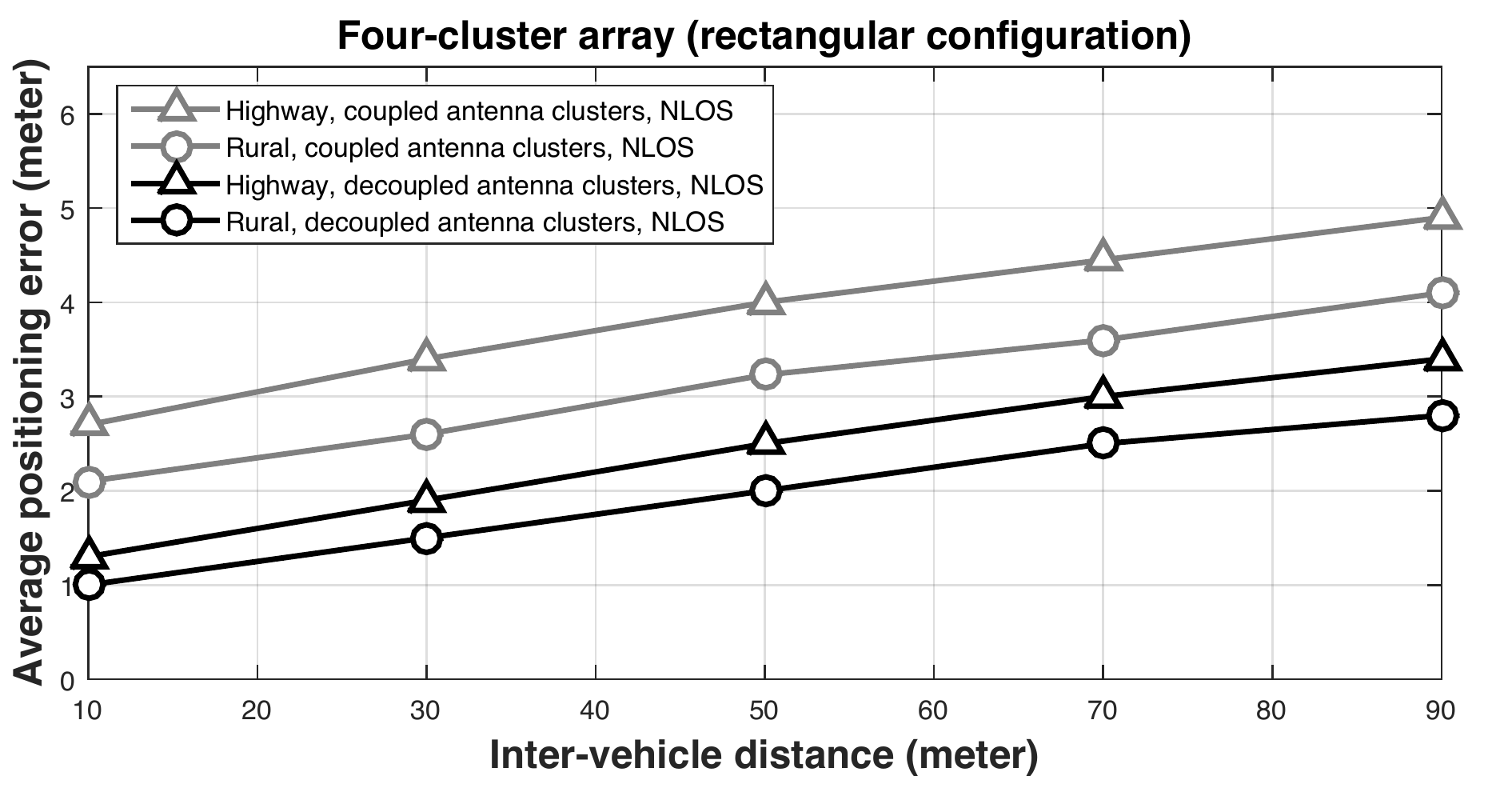}}
\vspace{-10pt}
\caption{Vehicular sensing accuracy under different number of signal paths and the inter-vehicle distance. }\label{positionError2D}
\vspace{-25pt}
\end{figure}

Fig.~\ref{positionError2D} shows the performance of the proposed vehicular sensing techniques in the  2D propagation model. The curves of average positioning error versus varying number of observed signal paths and inter-vehicle distance are plotted  in Fig.~\ref{positionError2D}(a) and \ref{positionError2D}(b), respectively. For the case of decoupled antenna clusters  using different orthogonal waveform sets, the LS estimator given in Section~\ref{sec:KarraysDifferent} is applied, which is feasible if  $P\geq6$ (see Proposition~\ref{pro:minNumPath_4MIMO}). On the other hand, for the case of coupled clusters, the technique of sensing box minimization  in Section~\ref{sec:SBM} is used requiring the number of observed paths $P\geq4$. Several key observations can be made as follows. First, from Fig.~\ref{positionError2D}(a),  receiving more observed paths at SV can dramatically decrease the positioning error and the positioning accuracy is significantly higher in the case of decoupled clusters than the other case.
Second, the positioning accuracy in the rural scenario is better than that in the highway scenario due to the following two reasons: 1) the signal propagation loss in the highway scenario is higher than that in the  rural counterpart as  typically longer  distance between vehicles and scatterers adds to  the difficulty of accurate sensing; 2) more paths exist   in the  rural scenario  due to  denser scatterers, which help improve the  positioning~accuracy.

Next, in Fig.~\ref{positionError2D}(b), it is shown that the average positioning error increases as the inter-vehicle distance grows  because larger path loss degrades sensing performance.  The gap of positioning error between the highway and rural scenarios increases with the inter-vehicle distance as the path loss scales up faster in the former than the latter. 

\begin{figure}[t]
\centering
\subfigure[Effect of portion of multi-bounce signal path.]{\includegraphics[width=8cm]{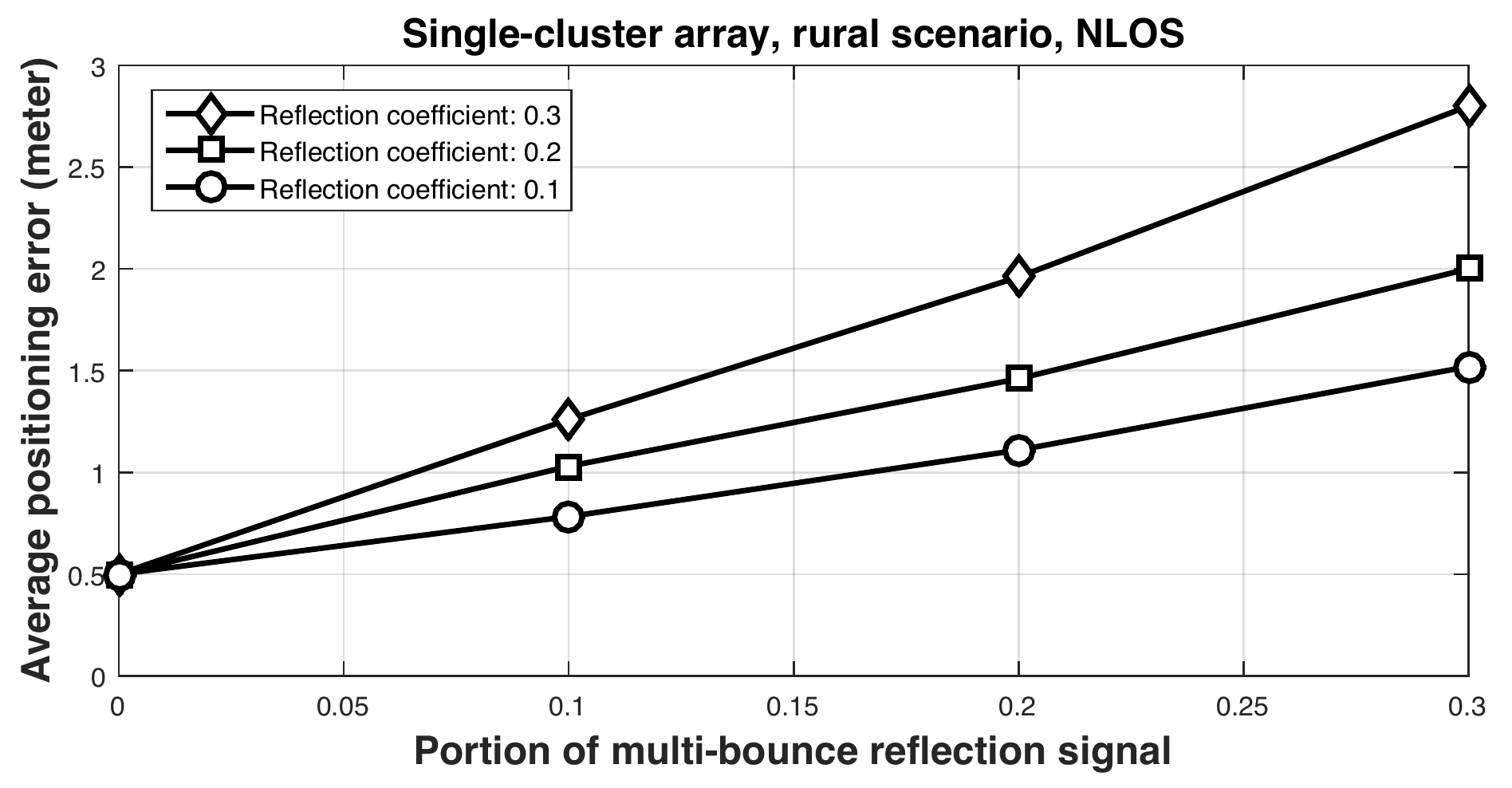}}
\subfigure[Effect of reflection coefficient.]{\includegraphics[width=8cm]{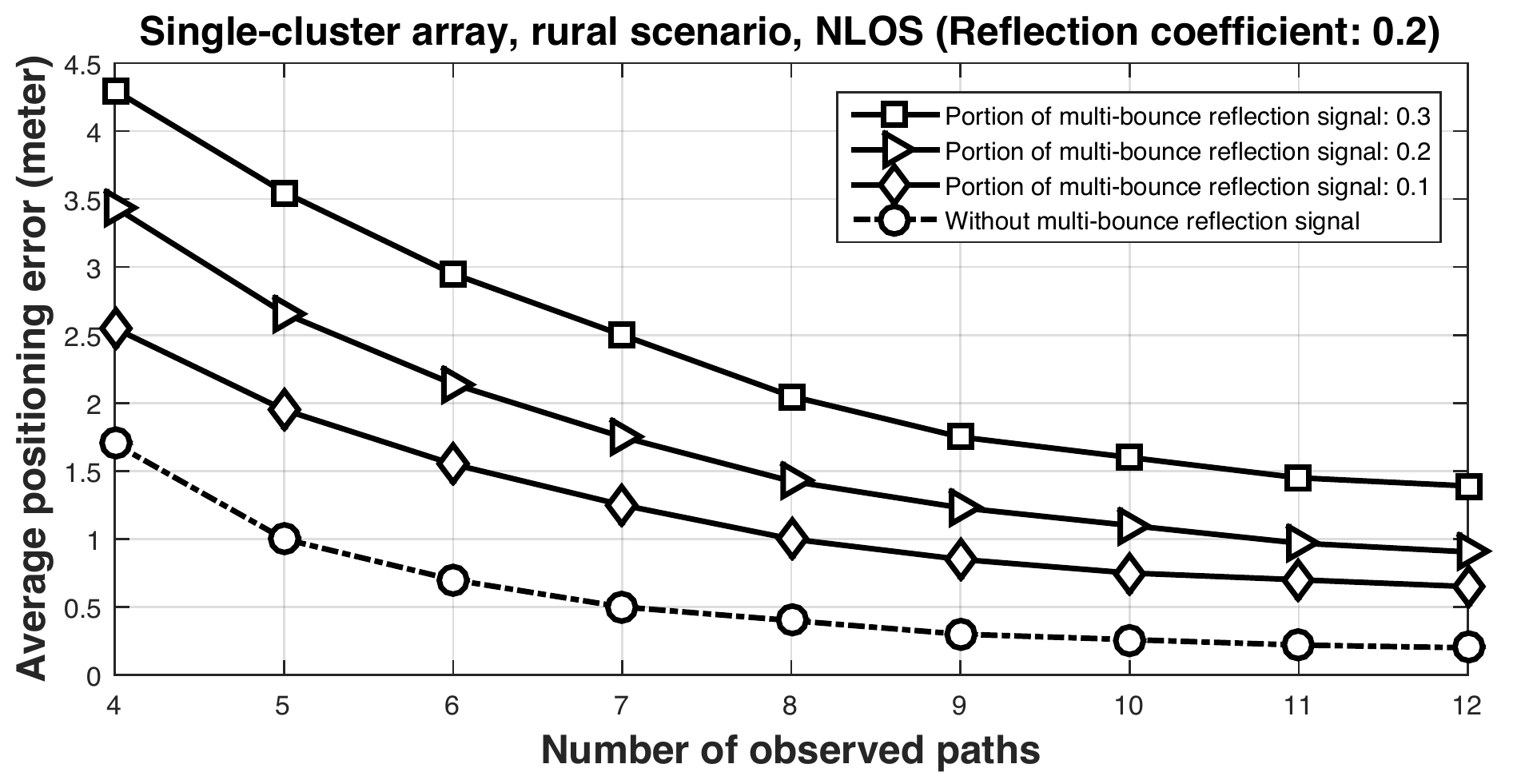}}
\vspace{-10pt}
\caption{Effects of the multi-bounce scattering on the positioning accuracy under different reflection coefficients. }\label{multiBounce}
\vspace{-30pt}
\end{figure}

Recall that our positioning technology is developed under the assumption of single-bounce scattering (see Assumption \ref{assump:SingleBounce}).
In rich scattering scenarios, however, signals are likely to reflect more than two bounces, defined as multi-bounce scattering, which results in the degradation of the positioning accuracy because the triangular geometry introduced in Section \ref{subsec:ScenarioDescription} is not satisfied any more. To investigate this effect, we consider an additional simulation scenario where single-bounce and multi-bounce scattering coexist. The multi-bounce signal paths are more attenuated than single-bounce paths
by multiplying the receive power with an additional reflection coefficient. Fig. \ref{multiBounce}(a)  presents the positioning error with respect to the fraction of multi-bounce signal paths, showing that the performance degradation is reduced as the additional reflection coefficient reduces. The reason is that multi-bounce paths usually have insufficient signal strengths, making the SV difficult to observe them. As a result, the dominant  observed paths are single-bounce signals from which the AoA, AoD, and ToA are accurately estimated, yielding the acceptable positioning performance with marginal degradation. Fig. \ref{multiBounce}(b) shows that the positioning error can be significantly decreased when more number of paths are observed. For instant, in case with the reflection loss of multi-bounce signals being $0.1$, the positioning error is less than $0.6$ m when $12$ paths are observed.

\vspace{-15pt}
\subsection{Comparison between 2D and 3D Propagations }
\vspace{-5pt}
Average positioning errors for 2D and 3D propagations are compared in Fig.~\ref{2D3D} for varying number of signal paths and inter-vehicle distance. We simulate the case that vehicle uses single-cluster array for transmission. First of all, the main trends  shown in Fig.~\ref{2D3D} are same as plotted in Fig.~\ref{positionError2D}. Specifically, according to  Fig.~\ref{2D3D}(a), the positioning for 2D and 3D propagation  are feasible when $P \geq 4$ and $P \geq 3$, respectively, aligned with Propositions \ref{pro:minNumPath} and \ref{pro:minNumPath3D}, respectively. Next, it is observed from both Fig.~\ref{2D3D}(a) and (b) that, compared with 2D propagation, the positioning accuracy for the 3D case   is worse and the error gap between the  highway and rural scenarios becomes larger. Due to relatively low scatterer height compared with  the inter-vehicle distance, most elevation angles are around $\frac{\pi}{2}$, and the resultant positioning accuracy tends to  be  sensitive to noisy angle detection. In addition, angle detection for 3D propagation  is  less accurate on the  highway  due to the larger propagation loss.

\begin{figure}[t]
\centering
\subfigure[Effect of number of observed signal paths.]{\includegraphics[width=8cm]{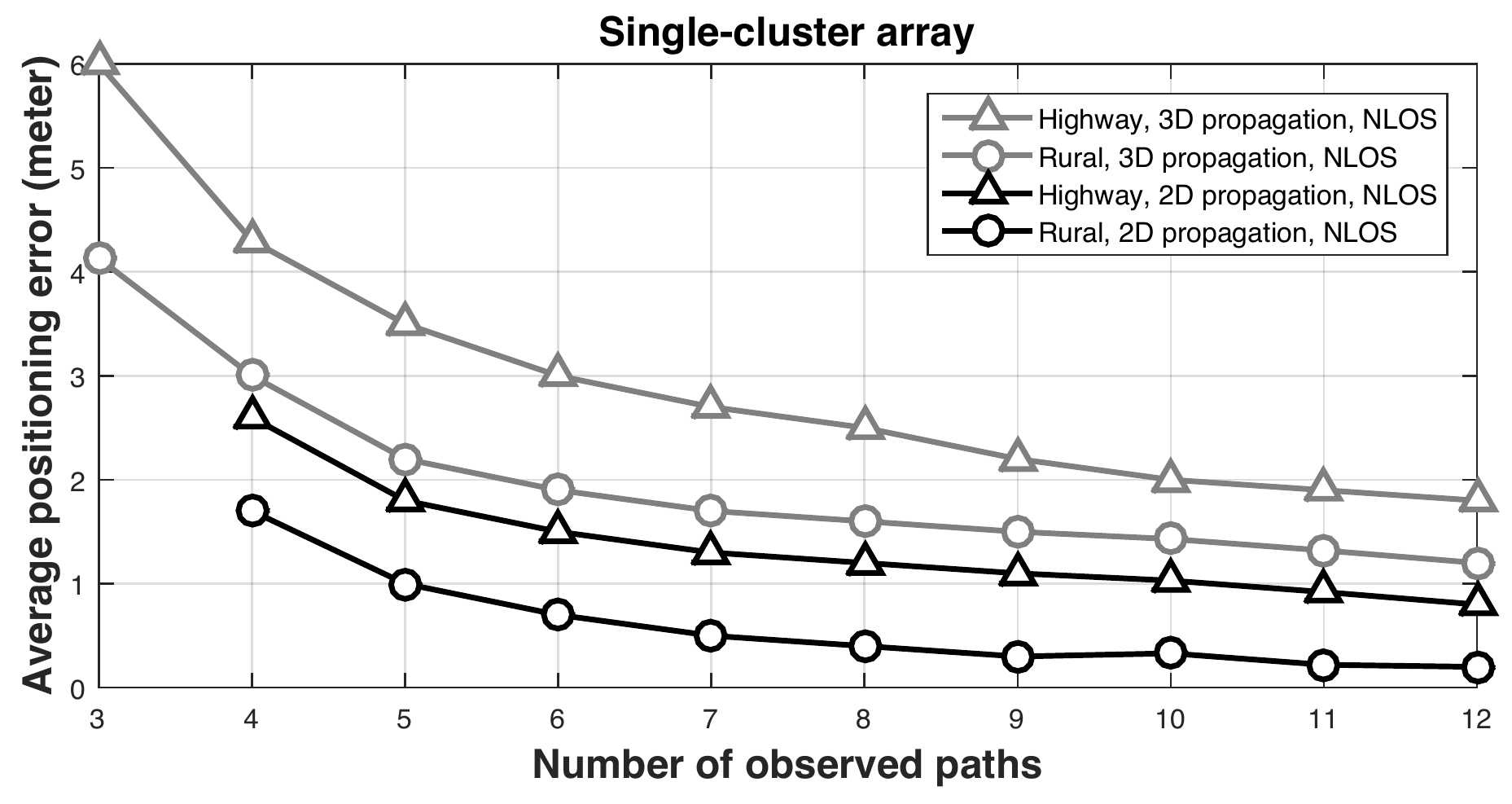}}
\subfigure[Effect of inter-vehicle distance.]{\includegraphics[width=8cm]{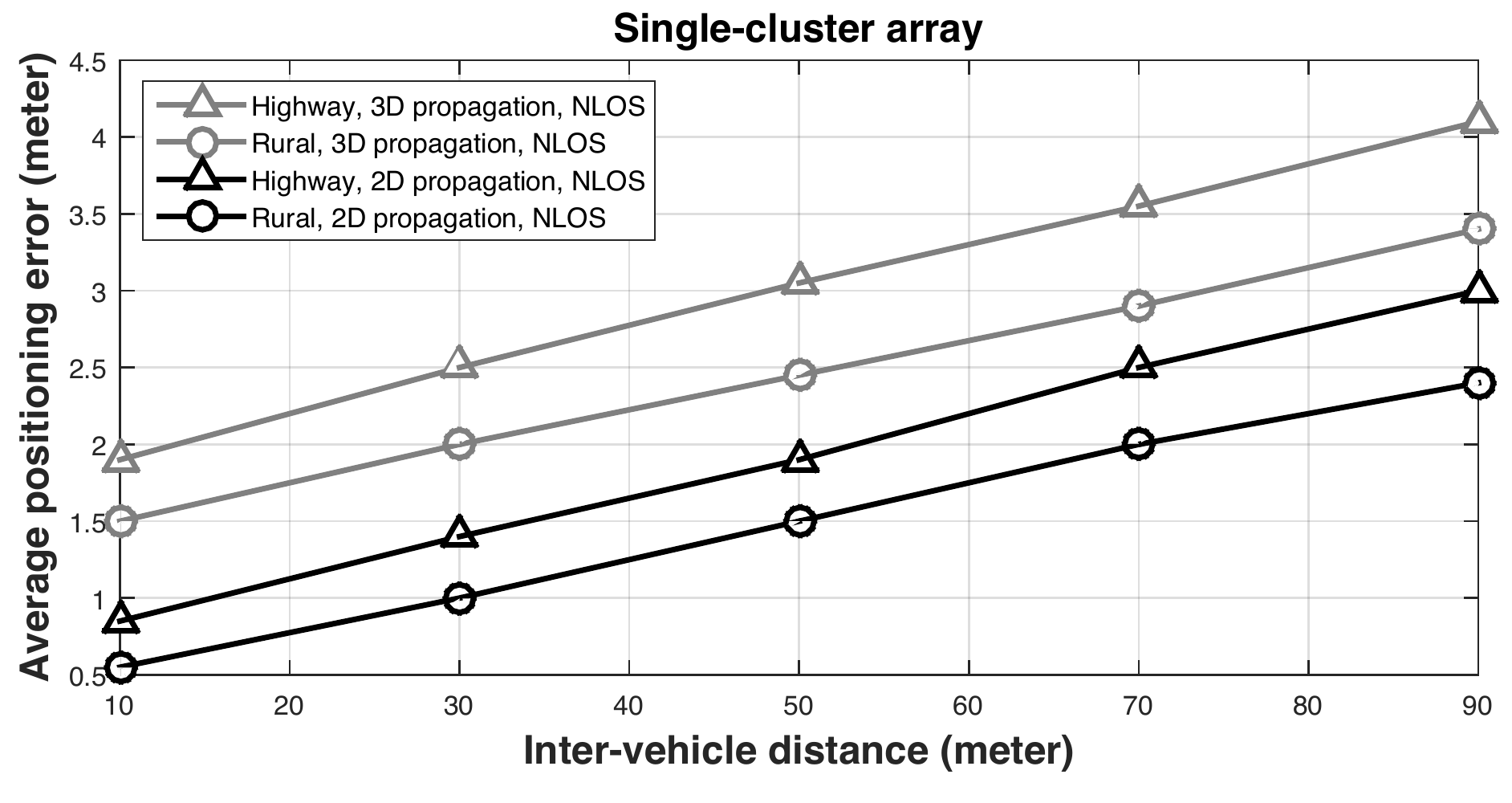}}
\vspace{-15pt}
\caption{Comparison of positioning accuracy between 2D and 3D propagations.}\label{2D3D}
\vspace{-20pt}
\end{figure}

Fig.~\ref{LoSComparison} verifies the feasibility of the proposed technique for sensing a vehicle in LoS condition. It is shown that the LoS path contributes to reducing positioning errors in both 2D and 3D propagation models due to its stronger signal strength than other NLoS paths, leading to more accurate TDoA/AoA/AoD detections as mentioned in Remark \ref{remark:LOS}. Fig. \ref{SNR} shows the effect of signal transmission power on positioning accuracy. It is observed that the positioning error can be reduced with higher transmission power because   more accurate TDoA/AoA/AoD detections are available.  Last, Fig. \ref{velocity} shows the effect of HV-SV relative velocity on positioning accuracy. Higher relative velocity results in larger Doppler shift to the signals. Therefore, the waveform orthogonality of signals might not be perfectly guaranteed, which leads to higher positioning error. Nevertheless, the positioning accuracy degradation is marginal within the practical velocity range (i.e., less  than 250 km/h).

\vspace{-18pt}
\subsection{Vehicular Size Sensing}\label{sim:Size}
\vspace{-7pt}

In the preceding results, we assume decoupled antenna clusters by using multiple orthogonal waveform sets for transmission. Next, we consider the case of coupled clusters and evaluate the performance of the vehicle size sensing techniques developed in Section~\ref{sec:KarraysSame}. To this end, we define a \emph{sizing error} as the area difference between the estimated size [sensing disk  in \eqref{sensingDisk} or sensing box  in \eqref{Eq:SizingBox}] and the real one. Fig. \ref{Sizing}(a) plots the average sizing errors of the disk and box minimizations versus the inter-vehicle distance.
The cases of underestimation and overestimation are separately measured, which is 
observed that the sizing errors of the both cases are increasing with inter-vehicle distance due to the same reason discussed for Fig.~9(b). Box minimization provides better sizing performance by exploiting the antenna clusters array configuration that is neglected by disk minimization. Fig. \ref{Sizing}(b) shows the both sizing methods are more likely to result in overestimation than underestimation, which is~aligned with the discussion in Remark~\ref{remark:SizeUnderestimation}.

\begin{figure}[t]
\centering
\begin{minipage}{0.495\textwidth}
\centering
\includegraphics[width=8cm]{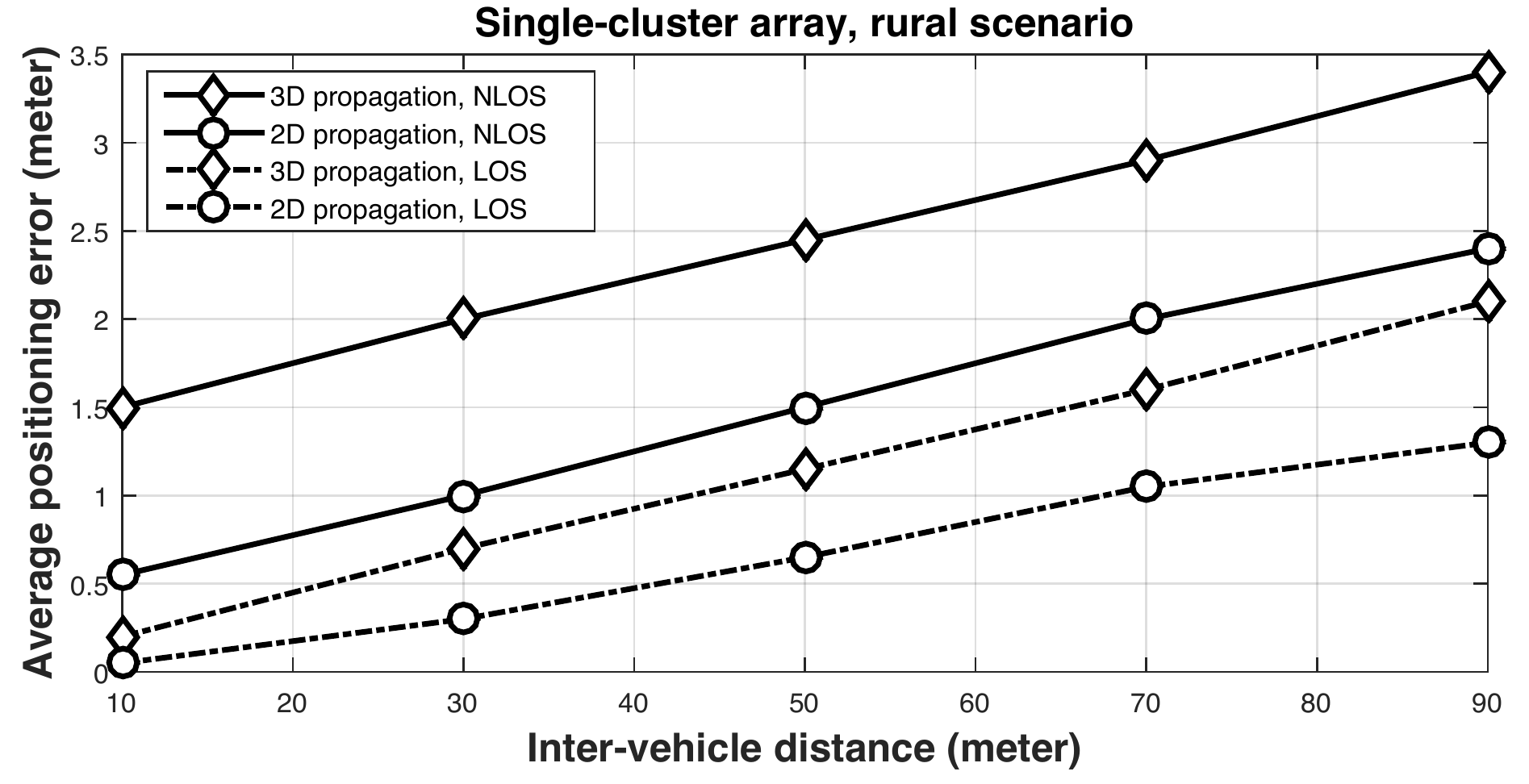}
\caption{Comparison of positioning accuracy between the cases with and without LoS path }\label{LoSComparison}
\end{minipage}
\begin{minipage}{0.495\textwidth}
\centering
\includegraphics[width=8cm]{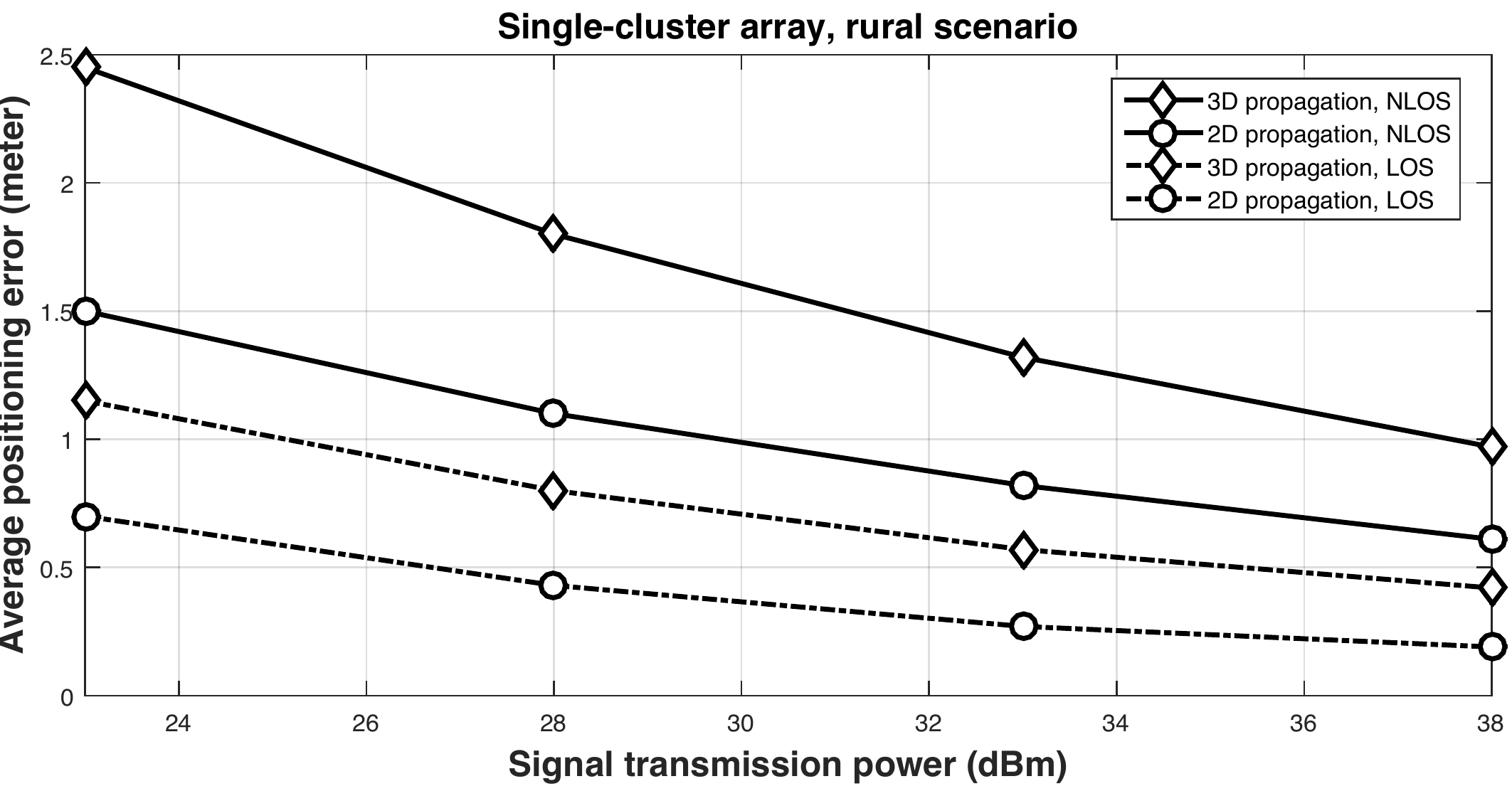}
\vspace{-10pt}
\caption{Effect of signal transmission power (i.e., SNR) on positioning accuracy.}\label{SNR}
\end{minipage}
\vspace{-20pt}
\end{figure}

\begin{figure}[t]
\centering
\includegraphics[width=8cm]{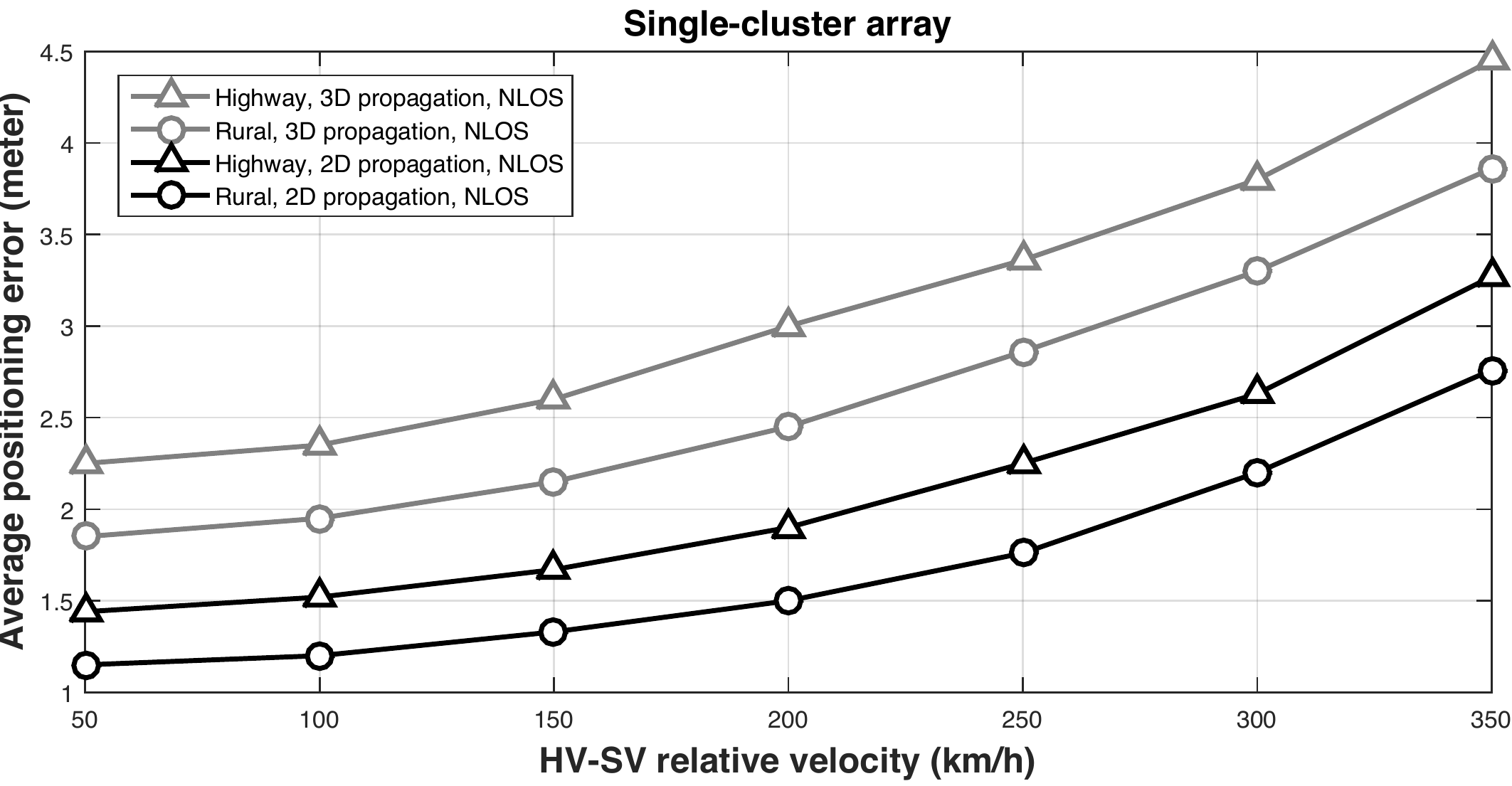}
\caption{Effect of HV-SV relative velocity on positioning accuracy.}\label{velocity}
\vspace{-20pt}
\end{figure}



\begin{figure}[t]
\centering
\subfigure[Effect of inter-vehicle distance on average sizing accuracy.]{\includegraphics[width=8cm]{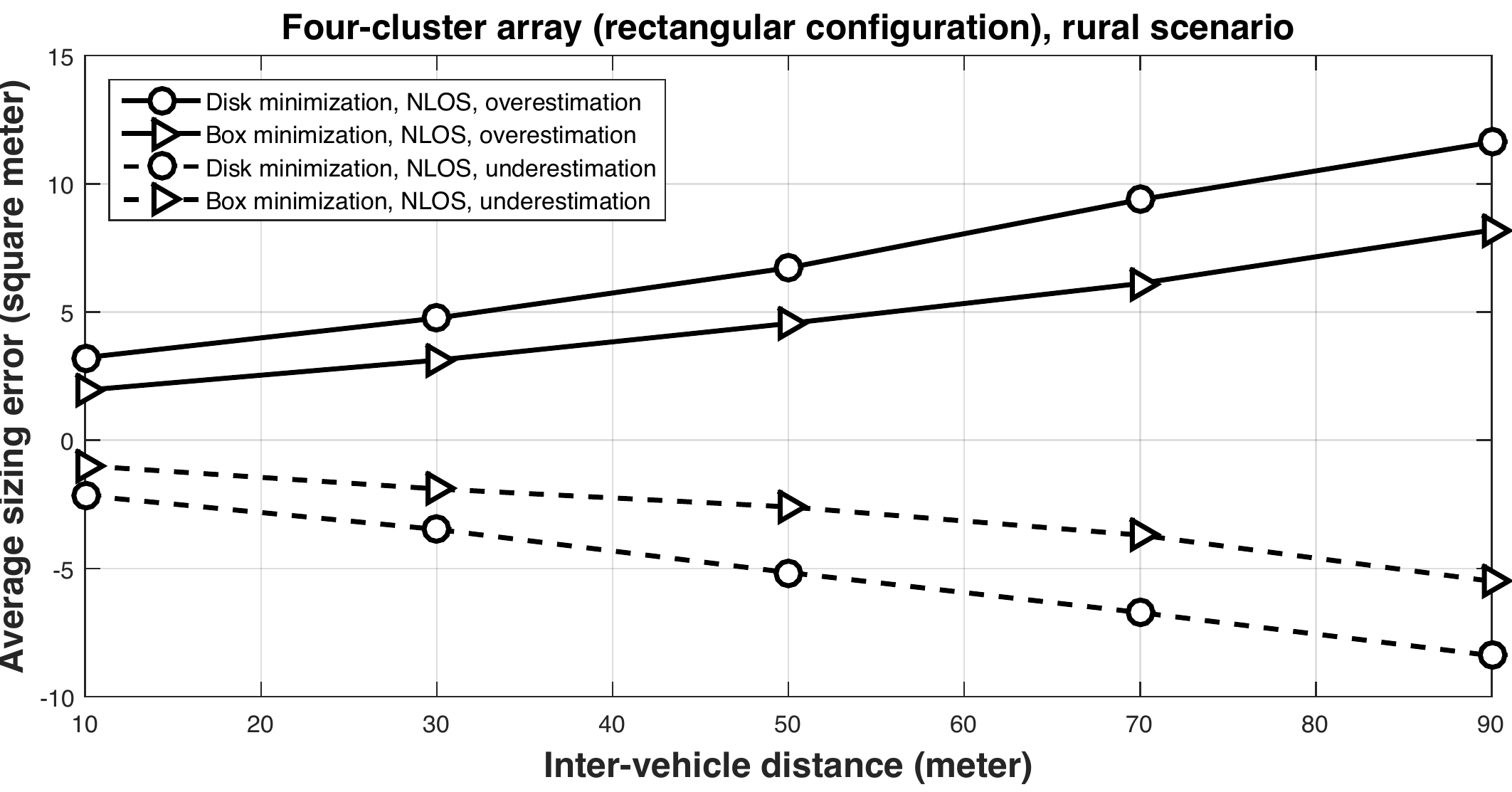}}
\subfigure[Effect of inter-vehicle distance on probability of overestimation of sizing.]{\includegraphics[width=8cm]{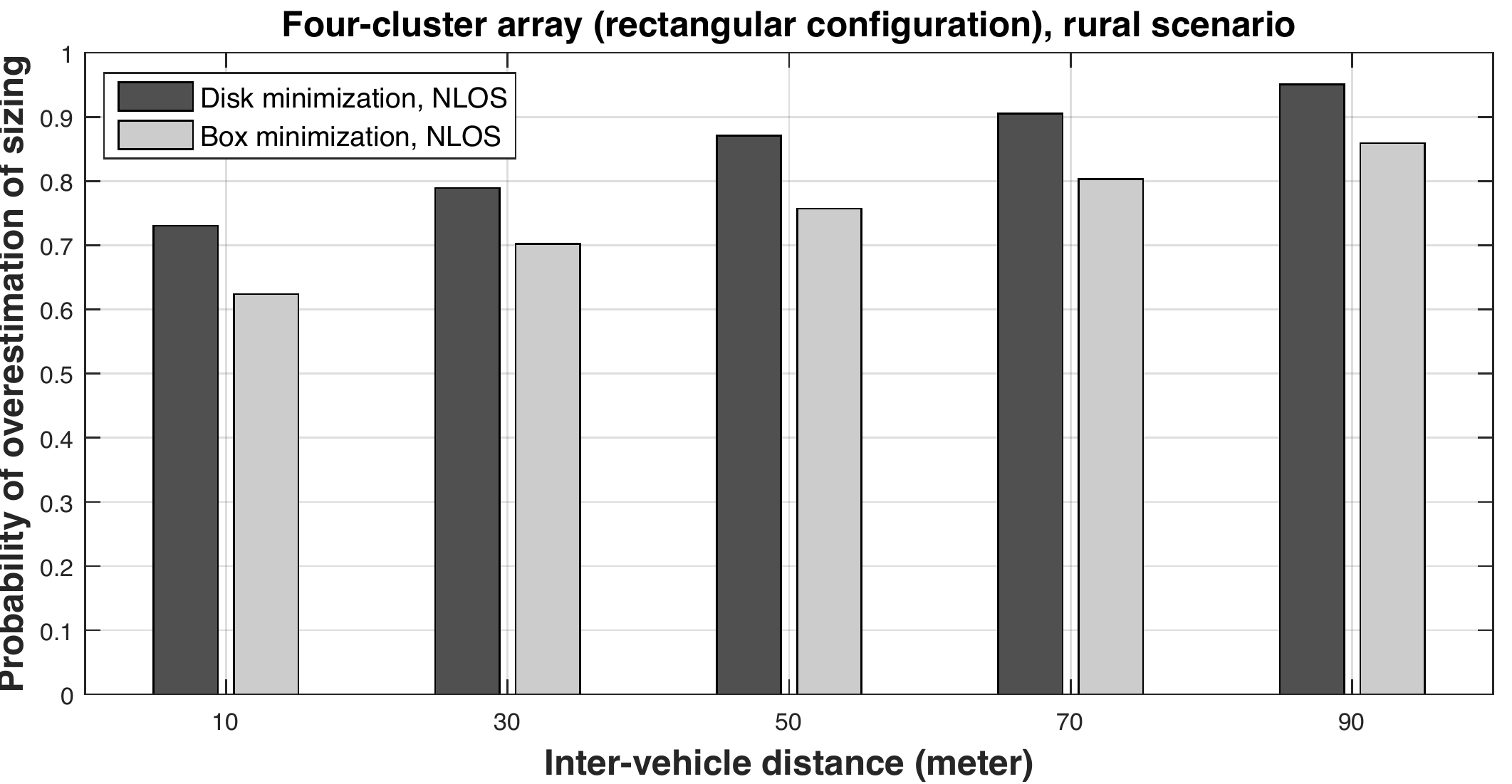}}
\vspace{-15pt}
\caption{Performance of proposed vehicular sizing approach.}\label{Sizing}
\vspace{-30pt}
\end{figure}

\vspace{-15pt}
\subsection{Coping with Insufficient Multi-Path}
\vspace{-5pt}
We evaluate the performance of the integrated solutions of  sequential path  combining and random directional beamforming both proposed in Section~\ref{sec:practicalIssue} in the case where scatterers are sparse and there are insufficient paths for vehicular sensing. The vehicle is equipped with a single cluster antenna array and transmits $Q$ repetitive signals with the transmission interval $\Delta=0.2~(\sec)$. The setting of $\Delta$ satisfies the two requirements explained in Sec. \ref{sec:SequentialCombine}, when considering the maximum coverage of V2V transmission $1$ km and assuming that the vehicle's  velocity is not changed within $2$ seconds.
We consider the random beamforming with the beam-width $\frac{2\pi}{Q}$ and $Q$ time instants sequential combining with $Q=\{1, 2, 4, 8\}$. Note that $Q=1$ corresponds to  isotropic transmission and serves as the benchmark.
The metric of  \emph{success sensing  probability} is defined as  the probability that the feasible condition of positioning is satisfied after applying sequential path  combining and random directional beamforming.
For example, in the case of isotropic transmission, at least $4$ paths should be detected at each time instant.
On the other hand, in case of random beamforming with sequential path combining, sensing  is feasible when the aggregate number of paths over  $Q$ time instants  should be larger than $5$ (see Section~\ref{sec:SequentialCombine}).
In Fig.~\ref{beamforming}(a), it is observed that the  sensing success  probability increases and approaches  to one as  $Q$ increases for both NLoS and LoS cases. Moreover, Fig.~\ref{beamforming}(b) shows that the positioning accuracy is improved
as sharper beam is used ($Q$ increases) since the angle detections become more accurate.

\vspace{-20pt}
\section{Concluding  Remarks}\label{sec:conclusion}
\vspace{-5pt}
This work presents the technologies for sensing a HV, namely its position, orientation, and size  by relying on V2V transmission and exploiting multi-path-geometry. Different techniques have been designed to support  tradeoffs   between the sensing accuracy and varying practical requirements in terms of   bandwidth,  signalling complexity, and array configuration. We have also addressed practical issue of insufficient multi-path for HV-sensing via developing effective techniques   exploiting path randomness in time domain. {The proposed technique can be used to assist the current vehicular sensing technologies  via giving the vehicle capability to sense HVs.} This work opens a new area of HV-sensing and points to many promising research topics on advanced HV-sensing such as velocity detection by estimating Doppler frequencies  and simultaneous sensing of multiple HVs.


\vspace{-20pt}

\begin{figure}[t]
\centering
\subfigure[Sensing success probability.]{\includegraphics[width=8cm]{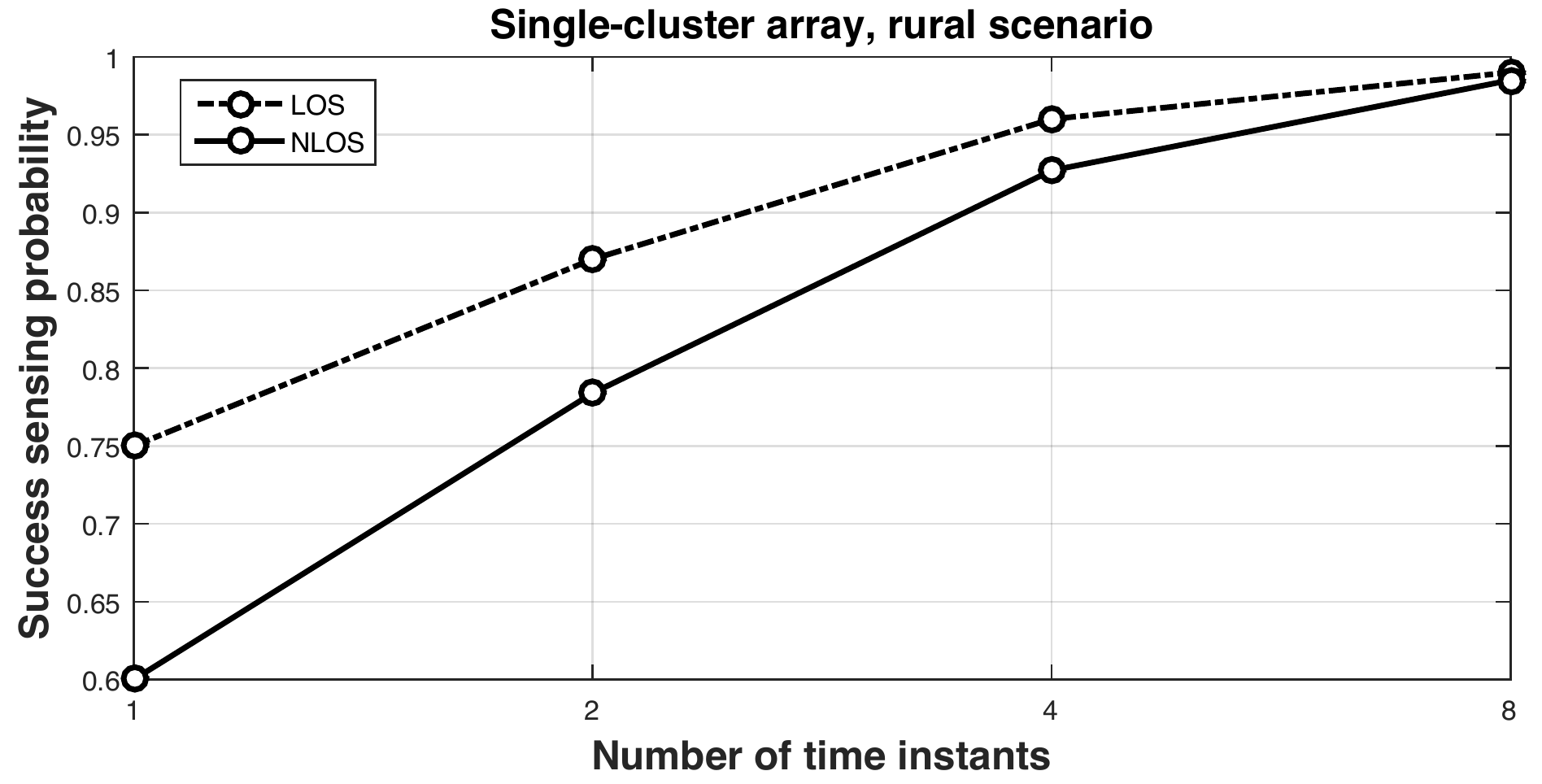}}
\subfigure[Effect of the inter-distance distance.]{\includegraphics[width=8cm]{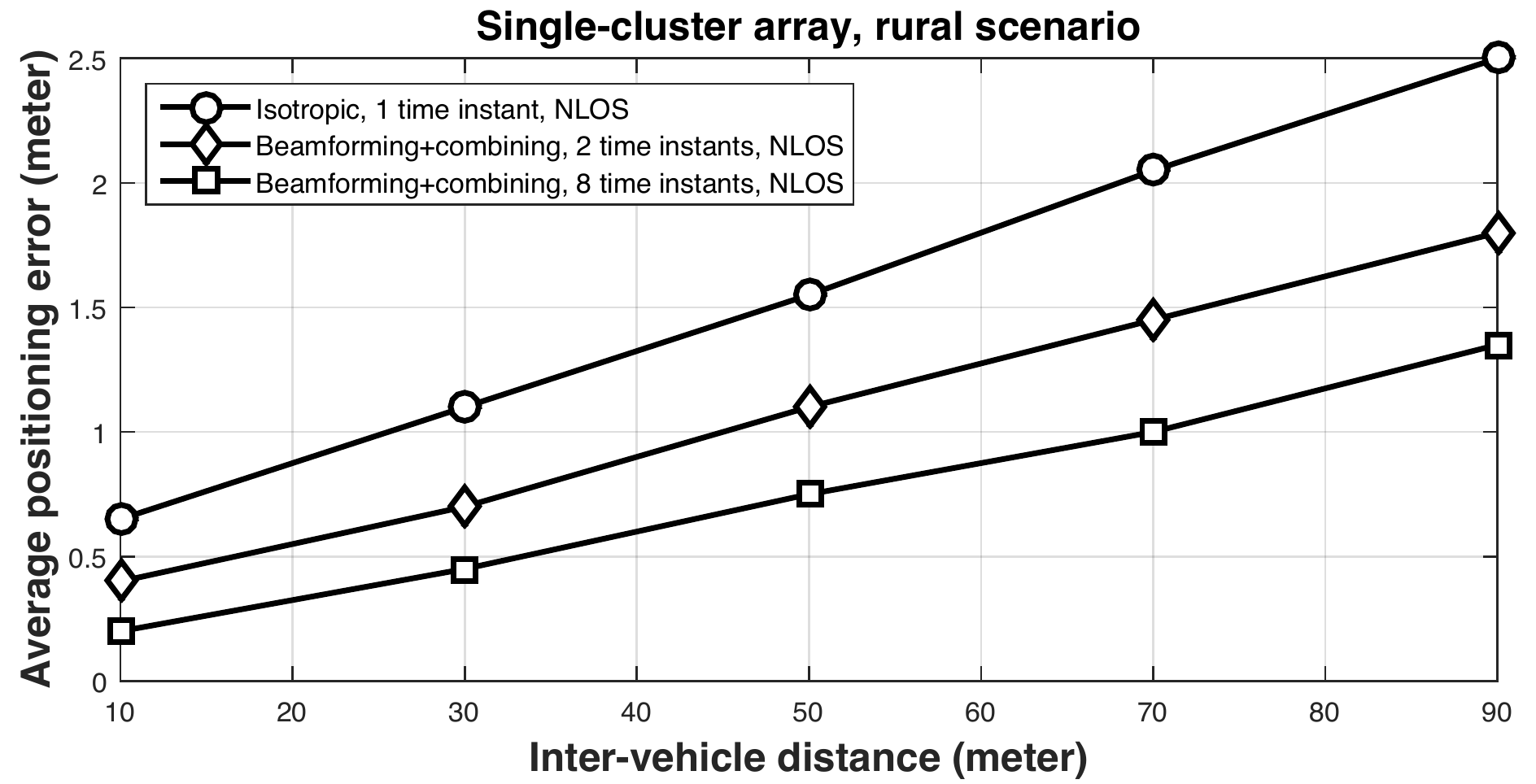}}
\vspace{-10pt}
\caption{Coping with  insufficient multi-path by sequential path combining and random directional beamforming. }\label{beamforming}
\vspace{-30pt}
\end{figure}

\vspace{-20pt}
\appendix
\vspace{-10pt}
\subsection{Proof of Proposition \ref{pro:optimaRD1_SCM}}\label{proof:Pro:minPosiCircle}
\vspace{-5pt}
The KKT conditions of Problem~ \ref{CircleMinProb} provide the following equality of the optimal solution~as
\begin{align}
&\sum_{p\in\mathcal{P}} \gamma_p\l(2\l(x_p - x_{0} \r)\frac{\partial x_p}{\partial d_1} + 2\l(y_p- y_{0} \r)\frac{\partial y_p}{\partial d_1}\r)=0,\label{KKT1_1}\\
&\gamma_p\l(2\l(x_p - x_{0} \r)\frac{\partial x_p}{\partial \nu_p} + 2\l(y_p- y_{0} \r)\frac{\partial y_p}{\partial \nu_p}\r)=0, \quad \forall p\in\mathcal{P}, \label{KKT1_2}\\
&\gamma_p \l(\l(x_p - x_{0} \r)^2 + \l(y_p- y_{0} \r)^2- r^2\r)=0, \quad \forall p\in\mathcal{P}, \label{KKT1_3}
\end{align}
where $\{\gamma_p\}$ represents the Lagragian multipliers of the first constraint in Problem \ref{CircleMinProb}. It is worth noting that at least two $\gamma_p$ should be strictly positive to satisfy \eqref{KKT1_1} and \eqref{KKT1_2}. From \eqref{KKT1_3}, it is obvious to lead \eqref{Optimal_Structure:SCM}, completing the proof.

\vspace{-20pt}
\subsection{Proof of Proposition \ref{pro:optimaRD1_SBM}}\label{proof:Pro:minPosiBox}
\vspace{-10pt}

The KKT conditions of Problem \ref{BoxMinProb} provide the following equality of the optimal solution as
\begin{align}
&\sum_{p\in\mathcal{P}}\Bigg[\bar{\gamma}_p\l(\cos(\omega) \frac{\partial x_p}{\partial d_1}+\sin(\omega)\frac{\partial x_p}{\partial d_1}\r)+\bar{\mu}_p\l(-\sin(\omega) \frac{\partial x_p}{\partial d_1}+\cos(\omega) \frac{\partial x_p}{\partial d_1}\r)\Bigg]=0,\label{KKT2_1}\\
&\Bigg[\bar{\gamma}_p\l(\cos(\omega) \frac{\partial x_p}{\partial \nu_p}+\sin(\omega) \frac{\partial x_p}{\partial \nu_p}\r)+\bar{\mu}_p\l(-\sin(\omega) \frac{\partial x_p}{\partial \nu_p}+\cos(\omega) \frac{\partial x_p}{\partial \nu_p}\r)\Bigg]=0,\label{KKT2_2}
\end{align}
where $\bar{\gamma}_p=\gamma_p^{(+)}-\gamma_p^{(-)}$ and $\bar{\mu}_p=\mu_p^{(+)}-\mu_p^{(-)}$ with Lagragian multipliers of the first constraint represented by $\gamma_p^{(+)}$, $\gamma_p^{(-)}$, $\mu_p^{(+)}$ and $\mu_p^{(-)}$, which are positive only when the corresponding equalities are satisfied. In other words, either $\gamma_p^{(+)} \l(\mu_p^{(+)}\r)$ or $\gamma_p^{(-)} \l(\mu_p^{(-)}\r)$ should be zero.

Some observations are made. First, to satisfy \eqref{KKT2_1} and \eqref{KKT2_2} simultaneously, at least two origins should be located in the boundary. Next, it is shown in \eqref{KKT2_2} such that if $\bar{\gamma}_p\neq 0$ then its counterpart multiplier $\bar{\mu}_p\neq 0$, which implies that the origin located in the boundary should be on the vertex. Last, the origin located at the vertex is equivalent to \eqref{Optimal_structure_SBM}, completing the proof.
\vspace{-10pt}

\subsection{Brief Introduction to Geometry-based Stochastic V2V Channel Model}

This section aims at summarizing some key features of our simulation channel model based on the geometry-based stochastic V2V channel model in \cite{karedal2009geometry}.   

\begin{enumerate}
\item Outline

The impulse response of the V2V channel, denoted by  $h(t)$, is expressed as the superposition of multiple signal paths, most of which are reflected by discrete scatterers including mobile and static ones as  
\begin{align}
h(t)=h_{\mathrm{LOS}}(t, \tau_{0})+\sum_{k=1}^{N_{\mathrm{MD}}} h_{\mathrm{MD}}(t,\tau_k)+\sum_{q=1}^{N_{\mathrm{SD}}} h_{\mathrm{SD}}(t,\tau_q), \nn
\end{align}
where $N_{{\mathrm{MD}}}$ and $N_{{\mathrm{SD}}}$ denote the numbers of mobile and statics scatterers (as described in Sec. VI), respectively, and $\tau$ refers to the time delay of the corresponding path. In case of NLoS/HV sensing, the term $h_{\mathrm{LOS}}(t, \tau_{0})$ is removed. Each signal path is given as 
\begin{align}
h(t, \tau)=\gamma \mathbf{b}\l(\theta\r) \mathbf{a}^{\mathrm{T}}\l(\varphi\r)\delta(t-\tau), \nn
\end{align}
where $\theta$ and $\varphi$ are the corresponding AoA and AoD, respectively, $\delta(t)$ is an impulse function of $t$, and $\gamma$ is the complex channel coefficient explained in the sequel. The  response vectors of HV and SV's antenna cluster $\mathbf{a}\l(\varphi\r)$ and $\mathbf{b}\l(\theta\r)$ are specified in (1) and (2), respectively. 

\item Layout and Scatterer Distribution

We consider a straight road of which the length is $600$ meters and the width is $8$ or $18$ meters depending on rural or highway scenarios \cite{karedal2009geometry}, respectively.   SV's location is fixed to the road's center and HV's location is changed according to the their distance. Mobile and static scatterers are randomly generated on the road and the road side, respectively. Their densities are specified in \cite[Table 1]{karedal2009geometry}. Given the locations of mobile and static scatterers, their delays $\{\tau\}$, AoAs $\{\theta\}$, and 
AoDs  $\{\varphi\}$ are directly calculated. 

\item Channel Coefficient

The complex channel coefficient $\gamma$ can be represented as the product of fast fading and path-loss. For fast fading, we consider a Rayleigh distribution with the mean specified in \cite{karedal2009geometry}. Next, consider a path-loss. Each signal path experiences different path-loss depending on its  flight distance and path-loss exponent. The path-loss exponent of a LoS path is fixed to $1.6$ for rural and $1.8$ for highway environments while that of a NLoS path is uniformly distributed as $\mathcal{U}[0, 3.5]$. 
\end{enumerate}

\bibliographystyle{ieeetr}

\end{document}

%% file: header.tex
\newtheorem{proposition}{Proposition}

\newtheorem{assumption}{Assumption}

\newtheorem{remark}{\bf Remark}
\def\phi{\varphi}

\def\l{\left}
\def\r{\right}
\def\({\left(}
\def\){\right)}

\def\bp{{\mathbf{p}}}

\def\b0{{\mathbf{0}}}

\def\bR{{\mathbf{R}}}

\newcommand{\nn}{\nonumber}